\documentclass[sigconf, nonacm]{acmart}
\AtBeginDocument{%
  \providecommand\BibTeX{{%
    \normalfont B\kern-0.5em{\scshape i\kern-0.25em b}\kern-0.8em\TeX}}}

\usepackage{amsmath, amssymb, amsthm}
\usepackage[mathscr]{euscript}
\usepackage{mathpartir}
\usepackage{graphicx}
\usepackage{mathtools, stmaryrd}
\usepackage[T1]{fontenc}
\usepackage{listings}
\usepackage{hyperref}
\usepackage{tikz}
\usepackage[normalem]{ulem}
\usepackage{caption}
\usepackage{subcaption}
\usepackage{version}

\usetikzlibrary{arrows.meta, positioning, shapes.geometric, calc}
\usepackage[scaled=0.75]{beramono} %

\definecolor{mygreen}{rgb}{0,0.6,0}
\lstdefinelanguage{links}{
  language=C,
  basicstyle={\linespread{0.8}\ttfamily},
  comment=[l]{\#\ },%
  escapeinside={(*}{*)},%
  tabsize=3,
  commentstyle=\color{mygreen},
  morekeywords={lens,select,join,drop,get,put,table,var,fun,database,where,with,tablekeys,from,determined,by,on,insert,into,values,typename,check,
  sig, spawn, spawnWait, redirect, receive, delete_left},
}
\newtheorem{lemma}{Lemma}
\newtheorem{theorem}{Theorem}
\newtheorem{defin}{Definition}
\newcounter{ProofCounter}
\newcommand*{\proofContext}[1]{\setcounter{equation}{0}\def\currentprefix{proof:#1}}
\newcommand*{\newProofContext}{\stepcounter{ProofCounter}\proofContext{num:\theProofCounter}}
\newcommand*{\loclabel}[1]{\label{\currentprefix:#1}}
\newcommand*{\locref}[1]{\ref{\currentprefix:#1}}

\renewcommand{\operatorname}[1]{\mathsf{#1}}

\tikzstyle{na} = [shape=rectangle, inner sep=0pt]
\tikzstyle{every picture} += [remember picture]
\tikzstyle{strike} = [color=red,thick]
\newcommand{\strikeStart}[1]{%
  \tikz[baseline=(begin.base)]\node[na](begin){#1};
}
\newcommand{\strikeEnd}[1]{%
  \tikz[baseline=(end.base)]\node[na](end){#1};
  \begin{tikzpicture}[overlay]
    \draw[strike] (begin.185) -- (end.5);
  \end{tikzpicture}
}
\newcommand{\strikeWord}[1]{%
  \tikz[baseline=(begin.base)]\node[na](begin){#1};
  \begin{tikzpicture}[overlay]
    \draw[strike] (begin.south west) -- (begin.north east);
  \end{tikzpicture}
}

\newcommand{\pipe}{\ |\ }
\newcommand{\extdef}{::=\ }

\newcommand{\correct}[2]{\strikeWord{#1}~\textcolor{red}{#2}}

\newcommand{\code}[1]{\lstinline{#1}}

\newcommand{\mkwd}[1]{\ensuremath{\mathsf{#1}}}
\newcommand{\op}[1]{\mkwd{#1}}
\newcommand{\kw}[1]{\textbf{#1}}

\newcommand{\evalsto}{\Downarrow}

\newcommand{\setfdsa}{\mathscr{F}}
\newcommand{\setfdsb}{\mathscr{G}}

 \newcommand{\oftype}[2]{#1 : #2}
\newcommand{\tyto}[3]{#1 \vdash #2 {\, : \,} #3} 
  \newcommand{\Tint}{\kw{int}} \newcommand{\Tbool}{\kw{bool}}
 
\newcommand{\Tstring}{\kw{string}} \newcommand{\Trecord}[1]{(#1)}
\newcommand{\Tunit}{\Trecord{}}

\newcommand{\vtrue}{\kw{true}} \newcommand{\vfalse}{\kw{false}}
 \newcommand{\set}[1]{\{ #1 \}}

\newcommand{\outputs}[1]{\op{outputs}(#1)}
 \newcommand{\ignoresoutputs}[2]{#1 \text{
    ignores } \outputs #2}

\newcommand{\domain}[1]{\mkwd{dom}(#1)}
\newcommand{\aliases}[1]{\op{dom}({#1})}

\newcommand{\recordproj}[2]{#1{.}#2}
\newcommand{\tyctx}{\Gamma} %
\newcommand{\typpa}{A}  %
\newcommand{\consta}{c} %
\newcommand{\laba}{\ell} \newcommand{\labb}{\ell'} \newcommand{\labc}{\ell''} %
\newcommand{\slaba}{\seq{\laba}} \newcommand{\slabb}{\seq{\labb}} \newcommand{\slabc}{\seq{\labc}} %

\newcommand{\rowa}{r} \newcommand{\rowb}{s} 
 
\newcommand{\expra}{M} \newcommand{\exprb}{N}
\newcommand{\preda}{P} \newcommand{\predb}{Q} %
\newcommand{\fdsa}{F} \newcommand{\fdsb}{G} %
\newcommand{\trowa}{R} \newcommand{\trowb}{R'} %
 
\newcommand{\rela}{S} \newcommand{\relb}{T} %
\newcommand{\schemaa}{\Sigma} \newcommand{\schemab}{\Delta} %
\newcommand{\lensa}{L} \newcommand{\lensb}{L'} %
\newcommand{\vala}{V}  %

\newcommand{\genop}{\odot} %

\newcommand{\tablee}[2]{\kw{table } #1 \kw{ with } #2}
\newcommand{\lens}[2]{\kw{lens } #1 \kw{ with } #2}
\newcommand{\lensselect}[2]{\kw{select}_{\lambda x.~#2} \kw{ from } #1}
\newcommand{\lensdrop}[4]{\kw{drop } #1 \kw{ determined by } (#2, #3) \kw{ from } #4}
\newcommand{\lensjoindl}[2]{\kw{join } #1 \kw{ with } #2 \kw{ delete\_left}}
\newcommand{\lensget}[1]{\kw{get } #1}
\newcommand{\lensput}[2]{\kw{put } #1 \kw{ with } #2}

\newcommand{\nodes}[1]{\op{names}(#1)}

\newcommand{\ruleref}[1]{\textsc{#1}}

\newcommand{\ifelse}[3]{\kw{if } #1 \kw{ then } #2 \kw{ else } #3}

\newcommand{\linksproj}[2]{#1{.}#2}

\let\oldcite\cite
\def\nobreakbefore{%
  \relax\ifvmode\else
    \ifhmode
      \ifdim\lastskip > 0pt\relax
        \unskip\nobreakspace
      \fi
    \fi
  \fi
}
\renewcommand\cite{\nobreakbefore\oldcite}

\newcommand{\sjf}[1]{{\noindent\small\color{blue} \framebox{\parbox{\dimexpr\linewidth-2\fboxsep-2\fboxrule}{\textbf{SJF:} #1}}}}

\newcommand{\rh}[1]{{\noindent\small\color{orange} \framebox{\parbox{\dimexpr\linewidth-2\fboxsep-2\fboxrule}{\textbf{RH:} #1}}}}

\newcommand{\jrc}[1]{\begin{center}{\noindent\small\color{red} \framebox{\parbox{\dimexpr\linewidth-2\fboxsep-2\fboxrule}{\textbf{JRC:} #1}}}\end{center}}

\def\disablecomments{}

\ifdefined\disablecomments
  \renewcommand{\sjf}[1]{}
  \renewcommand{\rh}[1]{}
  \renewcommand{\jrc}[1]{}
\fi

\newcommand{\totheleft}[1]{\begin{flushleft}#1\end{flushleft}}

\newcommand{\app}{\:}

\newcommand{\seq}[1]{\overrightarrow{#1}}

\newcommand{\midspace}{\: \pipe \:}

\newcommand{\secref}[1]{\S\ref{#1}}
\newcommand{\secrefp}[1]{(\secref{#1})}

\newcommand{\tablety}[2][\rela]{\kw{table of} \: (#1,#2)}
\newcommand{\lensty}[4][\schemaa]{\kw{lens of} \: (#1, #2, \lambda x.~#3, #4)}
\newcommand{\recordsetty}[1]{\kw{record set of} \: #1}
\newcommand{\bpvsort}[1]{\mkwd{sort}(#1)}

\newcommand{\lensid}{\kw{id}}%

\newcommand{\ltrans}[1]{\llparenthesis #1 \rrparenthesis}
\newcommand{\ltriple}[3][\schemaa]{#1 / #2 / #3}
\newcommand{\rlsort}[1]{\operatorname{sort} \left( #1 \right)}
\newcommand{\rldisjoint}[2]{\Sigma \uplus {#1} \Leftrightarrow \Sigma \uplus {#2}}
\newcommand{\Pset}[2]{{\operatorname{set}(\lambda x.~{#1}, #2)}}
\newcommand{\guid}[1]{#1 \text{ is globally unique}}
\newcommand{\subst}[3][x]{#2 [#3/#1]}
\newcommand{\ljd}[3]{\textbf{LJD}_{#2,#3}~ ({#1})}
\newcommand{\ljdi}[3]{\textbf{LJD}_{#2,#3}^\dagger~ ({#1})}
\newcommand{\defv}[3][\trowa,\trowb]{\textbf{DV}_{#1}~ ({#2})~{#3}}
\newcommand{\defvi}[2]{\textbf{DV}_{\trowa,\trowb}^\dagger~ ({#1})~{#2}}
\newcommand{\inh}[1]{\operatorname{inh}({#1})}
\newcommand{\rowapp}{\otimes}
\newcommand{\rowtyapp}{\oplus}
\newcommand{\sat}[2][\preda]{\operatorname{sat}(\lambda x.~{#1}, {#2})}

\newcommand{\norm}{\rightsquigarrow}
\newcommand{\opargs}[1]{\genop \{ #1 \}}
\newcommand{\opargshat}[1]{\hat{\genop} \{ #1 \}}

\newenvironment{fake}[1]{\par\vspace{3pt}\noindent\textbf{#1}\itshape}{\normalfont\ignorespacesafterend\vspace{3pt}\par}
\newenvironment{proofcase}[1]
  {\totheleft{\textbf{Case } #1}}
  {}

\newcommand{\bpvlensa}{I}
\newcommand{\bpvlensb}{J}
\newcommand{\setpreda}{\Pi}
\newcommand{\setpredb}{\Pi'}
\newcommand{\bpvselect}[3]{\texttt{select from } #1 \texttt{ where } #2 \texttt{ as } #3 }
\newcommand{\bpvdrop}[5]
  {\texttt{drop } #1 \texttt{ determined by } (#2, #3) \texttt{ from } #4
  \texttt{ as } #5}
\newcommand{\bpvjoin}[3]{\texttt{join\_dl } #1, #2 \texttt{ as } #3}
\newcommand{\bpvcompose}[2]{#1; #2}
\newcommand{\lamf}[2][x]{\lambda #1.~#2}
\newcommand{\lampreda}[1][x]{\lamf[#1]{\preda}}
\newcommand{\lampredb}[1][x]{\lamf[#1]{\predb}}

\newif\ifpublished

\publishedfalse

\ifpublished
\includeversion{published}
\excludeversion{techreport}

\acmConference[IFL'19]{International Symposium on Implementation and Application
of Functional Languages}{September 2019}{Singapore}
\acmYear{2020}
\copyrightyear{2020}
\setcopyright{acmlicensed}

\else
\renewcommand\footnotetextcopyrightpermission[1]{}
\excludeversion{published}
\includeversion{techreport}
\fi

\begin{document}

\ifpublished
\title{Language-Integrated Updatable Views}
\else
\title{Language-Integrated Updatable Views}
\subtitle{Extended version}
\fi
\author{Rudi Horn}
\orcid{0000-0001-6164-2208}
\affiliation{
  \institution{University of Edinburgh}            %
  \country{United Kingdom}                    %
}
\email{r.horn@ed.ac.uk}

\author{Simon Fowler}
\orcid{0000-0001-5143-5475}
\affiliation{
  \institution{University of Edinburgh}
  \country{United Kingdom}
}
\email{simon.fowler@ed.ac.uk}

\author{James Cheney}
\orcid{0000-0002-1307-9286}
\affiliation{
  \institution{University of Edinburgh \\ The Alan Turing Institute} %
  \country{United Kingdom}                    %
}
\email{jcheney@inf.ed.ac.uk}

\begin{abstract}
  \emph{Relational lenses} are a modern approach to the \emph{view update}
  problem in relational databases. As introduced
  by~\citet{bohannon2006relational}, relational lenses allow the definition of
  updatable views by the composition of lenses performing individual
  transformations.
  \citet{horn2018incremental} provided the first implementation of
  \emph{incremental relational lenses}, which demonstrated that
  relational lenses can be implemented efficiently by propagating
  \emph{changes} to the database rather than replacing the entire
  database state.

  However, neither approach proposes a concrete language design; consequently,
  it is unclear how to integrate lenses into a general-purpose programming
  language, or how to check that lenses satisfy the well-formedness conditions
  needed for predictable behaviour.
  In this paper, we propose the first full account of relational lenses in a
  functional programming language, by extending the Links web programming
  language. We provide support for higher-order predicates, and provide the
  first account of typechecking relational lenses which is amenable to
  implementation. We prove the soundness of our typing rules, and illustrate our
  approach by implementing a curation interface for a scientific database
  application.
  \end{abstract}

\maketitle

\section{Introduction}
\label{sec:introduction}
Relational databases are considered the \emph{de facto} standard for storing
data persistently, offering a ready-to-use method
for storing and retrieving data efficiently in a broad range of contexts.
\begin{figure}
    \texttt{albums} \vspace{0.75em}\\
  \begin{tabular}{|c|c|}\hline
    \texttt{album} & \texttt{quantity} \\ \hline
    Disintegration & 6 \\ \hline
    Show & 3 \\ \hline
    Galore & 1 \\ \hline
    Paris & 4 \\ \hline
    Wish & 5 \\ \hline
  \end{tabular}
  \vspace{1em} \\
  \texttt{tracks}\vspace{0.75em}\\
\begin{tabular}{|c|c|c|c|}\hline
  \texttt{track} & \texttt{year} & \texttt{rating} & \texttt{album} \\ \hline
  Lullaby & 1989 & 3 & Galore \\ \hline
  Lullaby & 1989 & 3 & Show \\ \hline
  Lovesong & 1989 & 5 & Galore  \\ \hline
  Lovesong & 1989 & 5 & Paris \\ \hline
  Trust & 1992 & 4 & Wish \\ \hline
\end{tabular}
  \caption{Music Database}
  \label{fig:music-db-example}
\end{figure}

Programs interface with relational databases using the \emph{Structured Query
Language} (SQL).  To query the database, the host application needs to generate
an SQL query from user input, issue it to the database server, and then process
the result in a way that aligns with the result of the query.

As an example, we consider a music database, originally proposed by~\citet{bohannon2006relational}
and shown in Figure~\ref{fig:music-db-example}. There are two tables: the
\texttt{albums} table, which details the quantities of albums available, and the
\texttt{tracks} table, which details the track name, year of release, rating, and
the album on which the track is contained.
Our application could generate a query by using string concatenation and then
assume the result will be in a known format containing records of track names of
type $\Tstring$ and years of type $\Tint$.

However, such an approach leaves many possible sources of error,
most of which are related to a lack of cross-checking of the different stages of
execution. The application could have bugs in query generation, which might
result in incorrect queries or even security flaws.
Furthermore, a generated query may not produce a result of the type that the
application expects, resulting in a runtime error.
The user experience of the programmer is also poor, as tooling provides little
help and the programmer must write code in two different languages, while being
mindful not to introduce any bugs in the application. We refer to this as an
\emph{impedance mismatch} between the host programming language and SQL
\cite{copeland1984making}.

Existing work on \emph{language integrated query} (LINQ) allows queries to be
expressed in the host language \cite{wong2000kleisli, cooper2006links}. Rather
than generating an SQL query using string manipulation, the query is written in
the same syntax as the host programming language. The user need not worry about
how the query is generated, and the code that performs the database query is
automatically type-checked at compile time.

As an example of LINQ, consider the following function, written in the
Links~\cite{cooper2006links}
programming language, which queries the \texttt{albums}
table and returns all albums with a given album name:

\begin{lstlisting}
fun getAlbumsByName(albumName) {
  for (a <-- albums)
    where (a.album == albumName)
    [a]
}
\end{lstlisting}%
The corresponding SQL for \lstinline+getAlbumsByName("Galore")+ would be:
\begin{center}
  \begin{tabular}{c}
    \begin{lstlisting}[language=SQL]
SELECT * FROM albums AS a WHERE a.album == "Galore"
    \end{lstlisting}
  \end{tabular}
\end{center}

LINQ approaches are convenient for querying databases, but still take a
relatively fine-grained approach to data manipulation (updates).  The programmer is
required to explicitly determine which changes were made at the application
level.  All modifications made by the user must then be translated into
equivalent insertions, updates and deletions for each table. In contrast, a
typical user workflow consists of fetching a subset of the database, called
a \emph{view}, making changes to this view, and then propagating the changes to the
database.
Defining views that can be updated directly is known as the \emph{view-update
problem}, a long-standing area of study in the field of databases
\cite{bancilhon1981update}.

\paragraph{Relational Lenses.}
A recent approach to the view update problem is to define views using composable
\emph{relational lenses} \cite{bohannon2006relational}. Lenses are a form of
\emph{bidirectional transformation}~\cite{FosterGMPS07:lenses}. With relational
lenses, instead of defining the view using a general SQL query, the programmer
defines the view by combining individual lenses, which are known to behave in a
correct manner. ~\citet{bohannon2006relational} define lenses for relational
algebra operations, in particular, projections, selections and joins.
Figure~\ref{fig:relational-lenses} shows the composable nature of relational
lenses.

A relational lens can be considered a form of
\emph{asymmetric} lens, in which we have a forward (\emph{get}) direction to fetch
the data, and a reverse (\emph{put}) direction to make updates
~\cite{hofmann2011symmetric}. A bidirectional transformation is
\emph{well-behaved} if it satisfies round-tripping guarantees:
\[
  \begin{array}{ll@{\qquad\:}ll}
    \ruleref{GetPut} & \op{put}\ s\ (\op{get}\ s) = s &
  \ruleref{PutGet} & \op{get}\ (\op{put}\ s\ v) = v
\end{array}
\]%
Relational lenses are equipped with typing rules which ensure operations on
lenses are well-behaved.  The type system for relational lenses tracks the
attribute types of the defined view as well as constraints, including
\emph{predicates} and \emph{functional dependencies}, which are not easily
expressible in an ML-like type system.

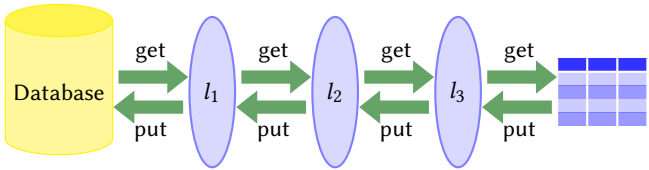
\begin{figure}

\def\arrshift{0.2}
\def\arrshiftup{(0,\arrshift)}
\def\arrshiftdn{(0,-\arrshift)}
\def\rowheight{0.18}
\def\rowwidth{1.2}
\def\rows{4}
\def\cols{3}

\begin{tikzpicture}
  [database/.style={cylinder, draw=yellow, cylinder uses custom fill,
    cylinder end fill=yellow!80, cylinder body fill=yellow!50, shape aspect=.3,
    shape border rotate=90, minimum height=2cm},
  lens/.style={ellipse, draw=blue!50, fill=blue!15, thick, minimum height=2cm},
  lArr/.style={-{Triangle[length=8pt,width=12pt]}, line width=5pt,shorten <=2, draw=black!60!green!60},
  thd/.style={fill=blue!20,width=4pt,height=1pt},
  font=\sffamily]

  \node[database] (db) {Database};
  \begin{scope}
    \node[lens, right=of db] (l1) {$l_1$};
    \node[lens, right=of l1] (l2) {$l_2$};
    \node[lens, right=of l2] (l3) {$l_3$};
  \end{scope}
  \draw[lArr] ($(db.0)+(0,\arrshift)$) -- ($(l1.180)+(0,\arrshift)$) node[midway, above] {get};
  \draw[lArr] ($(l1.180)-(0,\arrshift)$) -- ($(db.0)-(0,\arrshift)$) node[midway, below] {put};

  \draw[lArr] ($(l1.0)+(0,\arrshift)$) -- ($(l2.180)+(0,\arrshift)$) node[midway, above] {get};
  \draw[lArr] ($(l2.180)-(0,\arrshift)$) -- ($(l1.0)-(0,\arrshift)$) node[midway, below] {put};

  \draw[lArr] ($(l2.0)+(0,\arrshift)$) -- ($(l3.180)+(0,\arrshift)$) node[midway, above] {get};
  \draw[lArr] ($(l3.180)-(0,\arrshift)$) -- ($(l2.0)-(0,\arrshift)$) node[midway, below] {put};

  \coordinate[right=of l3] (tbl) {};

  \begin{scope}[shift={($(tbl)-(0,{(\rows+1)*\rowheight/2})$)}, line/.style={color=white, thick}, local bounding box=tbl]
    \foreach \x in {0,...,\rows} {
      \ifodd \x
        \def\fopac{0.2}
      \else 
        \def\fopac{0.4}
      \fi
      \ifnum \x=\rows
        \def\fopac{0.8}
      \fi
      \fill[color=blue,opacity=\fopac] (0,\x*\rowheight) rectangle (\rowwidth,{(\x+1)*\rowheight});
    }
    \draw[line] (0,\rows*\rowheight) -- (\rowwidth,\rows*\rowheight);
    \foreach \x in {0,...,\cols} {
      \draw[line] ({\x*(\rowwidth/\cols)},0) -- ({\x*(\rowwidth/\cols)},{(\rows+1)*\rowheight});
    }
  \end{scope}

  \draw[lArr] ($(l3.0)+(0,\arrshift)$) -- ($(tbl.180)+(0,\arrshift)$) node[midway, above] {get};
  \draw[lArr] ($(tbl.180)-(0,\arrshift)$) -- ($(l3.0)-(0,\arrshift)$) node[midway, below] {put};
\end{tikzpicture}

   \caption{Relational Lenses}
  \label{fig:relational-lenses}
\end{figure}

\paragraph{From theory to practice.}
The theory of relational lenses was developed over a decade ago
by~\citet{bohannon2006relational}, but until recently there has been little work
on practical implementations. \citet{horn2018incremental} recently presented the first
implementation using an incremental semantics.
However~\citet{horn2018incremental} focus on performance rather than language
integration, leaving two issues unresolved:
\begin{itemize}
\item How to integrate relational lenses, which are defined as a sequential
  composition of primitives, into a functional language, where lenses are
  composed using lens subexpressions.
\item How to define and verify the correctness of a concrete \emph{selection
    predicate} syntax for relational lenses.
\end{itemize}

\paragraph{Predicates.}
Some of the relational lens constructors, such as the \emph{select} lens, require user
supplied functions for filtering rows. Such functions, called \emph{predicates},
determine whether or not an individual record should be included. Predicates are
a function of type $\trowa \to \Tbool$ where $\trowa$ is the type of the
input record and a return value of $\vtrue$ indicates that the predicate holds.

\citet{bohannon2006relational} treat predicates as abstract (finite or infinite)
sets, without giving a computational syntax. Sets allow predicates to be defined
in an abstract form while still being amenable to mathematical reasoning, but
such an approach does not scale to a practical implementation. In practice the
user should define a predicate as a function from a record (in this case
containing \lstinline+album+ and \lstinline+year+ fields) to a Boolean value:
\begin{lstlisting}
fun(x) { x.album = "Galore" && x.year == 1989 }
\end{lstlisting}
Some of the lens typing rules require static checks on predicates. The above
predicate contains only \emph{static} information, and is thus
a \emph{closed function} which can be checked at compile-time. We call such
predicates \emph{static} predicates.
Alas, such checks become problematic when the programmer would like to define a
function which depends on information only available at runtime, such as a
parameter in a web request. For example, consider the following function which
adapts the \lstinline+getAlbumsByName+ function to use relational lenses.
\begin{lstlisting}
fun getAlbumsByNameL(albumName) {
  var albumLens = lens albums where album -> quantity;
  var selectLens = select from albumsLens where
    fun(a) { a.album == albumName };
  get selectLens;
}
\end{lstlisting}

The \lstinline+getAlbumsByNameL+ function begins by defining
\lstinline+albumLens+ as a lens over the \lstinline+albums+ table. A
\emph{functional dependency} ${\scriptstyle \slaba \to \slabb}$ states that the
columns in ${\scriptstyle \slabb}$ are \emph{uniquely determined} by
${\scriptstyle \slaba}$; here, the \lstinline+album -> quantity+ clause states
that the \lstinline+quantity+ attribute is uniquely determined by the
\lstinline+album+ attribute.

As \lstinline+albumsName+ is supplied as a parameter to the
\lstinline+getAlbumsByNameL+ function, the anonymous predicate supplied to
\lstinline+select+ can only be completely known at runtime. We call such
predicates \emph{dynamic} predicates. Dynamic predicates are not closed, which
means that variables in the closure of a dynamic predicate may not be available
until runtime, and may themselves refer to functions.
While it is possible to statically know the \emph{type} of the function, and
thus rule out a class of errors, relational lenses require finer-grained checks
which require a more in-depth analysis of the predicate. As an example, a select
lens is only well-formed if the predicate does not rely on the output of a
functional dependency.

If we required the function to be fully known at compile time, a
programmer could not define predicates that depend on user input.
Thus, there is a tradeoff between static correctness and programming flexibility.
In our design, we can perform checks on lenses using static predicates at
compile-time, and we can also support dynamic predicates by performing the
same checks at runtime.

Another obstacle is the handling of functional dependencies, which are
an important part of the type system for relational lenses. Functional
dependencies are constraints that apply to the data, and specify which
fields in a table
uniquely determine other fields.

The typing rules given by~\citet{bohannon2006relational} are important for
showing soundness of relational lenses: without ensuring all the requirements
are met, it is not possible to ensure the lenses are well-behaved. We take the
existing work by~\citet{bohannon2006relational} and concretise and adapt
the design to allow the rules to be implemented in practice.

\subsection{Contributions}
The primary technical contribution of this paper is the first full
design and implementation of relational lenses in a typed functional programming
language, namely Links~\cite{cooper2006links}.
This paper makes three concrete contributions:

\begin{enumerate}
  \item A design and implementation of \emph{predicates} for relational lenses,
    based on previous approaches to language-integrated query.  We define a
    language of predicates, and show how terms can be normalised to a fragment
    both amenable to typechecking of relational lenses, and translation to SQL.
  \item An implementation of the typing rules for relational lenses, adapted to
    the setting of a functional programming language \secrefp{sec:impl:tc}.
    We prove~\secrefp{sec:impl:tc:correctness} that our compositional typing
    rules are sound with respect to the original rules proposed
    by~\citet{bohannon2006relational}.
    Static predicates can be fully checked at compile time, whereas the same
    checks can be performed on dynamic predicates at
    runtime.%
  \item A curation interface for a real-world scientific database implemented as
    a cross-tier web application, tying together relational lenses with the
    Model-View-Update architecture for frontend web
    development~\secrefp{sec:case-study}.
\end{enumerate}

We have packaged our implementation and example application as an
artifact~\cite{HornFC20:artifact}.
\begin{published}
Proofs of the technical results can be found in the extended version of the
paper~\cite{HornFC20:extended}.
\end{published}

The remainder of the paper proceeds as follows:~\secref{sec:impl:predicates}
describes the design and implementation of predicates;~\secref{sec:impl:tc}
describes the implementation of static typechecking for relational
lenses;~\secref{sec:case-study} describes the case study;~\secref{sec:related}
describes related work; and~\secref{sec:conclusion} concludes.

\section{Predicates}\label{sec:impl:predicates}
In their original proposal for relational lenses,~\citet{bohannon2006relational}
define predicates using abstract sets. Although theoretically convenient, such a
representation is not suited to implementation in a programming language.
Our first task in implementing relational lenses, therefore, is to define a
concrete syntax for predicates.

As we are working in the setting of a functional programming language, it is
natural to treat predicates as  functions from  records to  Boolean values.
As an example, recall our earlier example of the \lstinline+select+ lens,
which selects albums with a given name:
\begin{lstlisting}
select from albumsLens where fun(a) { a.album == albumName }
\end{lstlisting}
Here, the predicate function is
\lstinline+fun(a) {a.album == albumName}+.
Intuitively, this predicate includes a record \lstinline+a+ in the set of results if its
\lstinline+album+ field matches \lstinline+albumName+.

In our approach, predicates are
a well-behaved subset of Links functions which take a parameter of the type of
row on which the lens operates. We define a simply-typed $\lambda$-calculus for
predicates, and apply the normalisation approach advocated
by~\citet{Cooper09:linq-norm} to derive a form which is both amenable to SQL
translation, and can be used when typechecking lens construction.

\subsection{Static and Dynamic Predicates}

Ensuring relational lenses are well-typed requires some conditions that require
static knowledge of predicates. As an example, we require that
the predicate of a \lstinline+select+ lens does not refer to the outputs of the
functional dependencies of a table; we describe the conditions more in detail in
Section~\ref{sec:impl:tc}.

\sloppypar
Our approach distinguishes two types of predicates: \emph{static} predicates,
which rely on only static information; and \emph{dynamic} predicates, which
can refer to arbitrary free variables.
Referring to our previous example,
\lstinline+fun(a) { a.album == albumName }+ is a dynamic predicate, as
\lstinline+albumName+ is a free variable, whereas \lstinline+fun(a) { a.album == "Paris" }+
is a static predicate.

We can check the construction of lenses with static predicates entirely
statically, whereas lenses with dynamic predicates require the same checks to be
performed dynamically. Our formal results are based on static predicates,
however the same results apply for dynamic predicates (which can be treated as
closed at runtime).

\subsection{Predicate Language}

\begin{figure}
  ~\totheleft{Syntax}
\[
  \begin{array}{lrcl}
    \text{Types} & A, B, C & \extdef & A \to B \midspace \Trecord{\seq{\laba :
    A}} \midspace D \\
    \text{Base types} & D & \extdef &  \Tbool \midspace \Tint \midspace \Tstring \\
    \text{Base record types} & \trowa & \extdef & \Trecord{\seq{\laba : D}} \\ \\
    \text{Labels} & \laba \\
    \text{Terms}  & L, M, N & \extdef & x \midspace \consta \midspace \lambda
    x.~M \midspace M \app N \\
     & & \midspace & \Trecord{\seq{\ell = M}} \midspace \linksproj{M}{\laba} \\
     & & \midspace & \ifelse{L}{M}{N} \\
     & & \midspace & \opargs{\seq{M}} \\
  \end{array}
\]

~\totheleft{Typing rules}

\begin{mathpar}
  \inferrule
  [T-Var]
  { x : A \in \Gamma }
  { \Gamma \vdash x : A }

  \inferrule
  [T-Const]
  { c \text{ of type } A }
  { \Gamma \vdash c : A }

  \inferrule
  [T-Abs]
  { \Gamma, x : A \vdash M : B }
  { \Gamma \vdash \lambda x . M : A \to B}

  \inferrule
  [T-App]
  { \Gamma \vdash M : A \to B \\ \Gamma \vdash N : A }
  { \Gamma \vdash M \app N : B }

  \inferrule
  [T-Record]
  { (\Gamma \vdash M_i : A_i )_i \\\\
  \text{ for each } M_i : A_i \in \seq{M : A} }
  { \Gamma \vdash \Trecord{\seq{\ell = M}} : \Trecord{\seq{\ell : A}} }

 \inferrule
  [T-Project]
  { \Gamma \vdash M : \Trecord{\laba_i : A_i}_{i \in I} \\ j \in I  }
  { \Gamma \vdash \recordproj{M}{\laba_j} : A_j }

  \inferrule
  [T-If]
  { \Gamma \vdash L : \Tbool \\\\ \Gamma \vdash M : A \\ \Gamma \vdash N : A }
  { \Gamma \vdash \ifelse{L}{M}{N} : A }

  \inferrule
  [T-Op]
  { \genop : D_1 \times \ldots \times D_n \to D \\
    (\Gamma \vdash M_i : D_i)_{i \in 1..n}
  }
  { \Gamma \vdash \opargs{\seq{M}} : D}
\end{mathpar}

\caption{Syntax and typing rules for predicate language}
\label{fig:syntax}
\end{figure}

\paragraph{Syntax.}
Figure~\ref{fig:syntax} shows the syntax of the predicate language.
Types, ranged over by $A, B, C$, include function types $A \to B$; record types
${\scriptstyle \Trecord{\seq{\laba ~ : ~ A}}}$ mapping labels $\ell$ to values of type $A$; and
base types $D$, ranging over the types of Boolean values, strings, and
integers.
It is convenient to let $\trowa$ range over records whose fields are of base
type. The unit type $\Tunit$ is definable as a record with no fields.

Terms, ranged over by $L, M, N$, are those of the simply-typed
$\lambda$-calculus extended with base types, records, conditional statements,
and $n$-ary operators on base types ${\scriptstyle \opargs{\seq{M}}}$. We assume that the set
of available operators all have an SQL equivalent and assume the existence of at
least the comparison operators $<, >, ==$ and Boolean negation, conjunction, and
disjunction. We sometimes find it convenient to use infix notation for
binary operators.

\paragraph{Typing.}
Most typing rules are standard for the simply-typed $\lambda$-calculus extended
with records.  The only non-standard rule is \textsc{T-Op}, which states that
the arguments to an operator must be of base type and match the type of the
operator.

\paragraph{Normalisation.}

\begin{figure}
\totheleft{Normal forms}
\[
  \begin{array}{rcl}
    O & ::= & x \midspace c \midspace \lambda x . O \midspace (\seq{\ell = O}) \midspace x.\ell \\
      & \midspace & \ifelse{O_1}{O_2}{O_3} \midspace \opargs{\seq{O}} \\
    P, Q & ::= & \ifelse{P_1}{P_2}{P_3}
      \midspace \opargs{\seq{P}} \midspace x.\ell \midspace c
  \end{array}
\]

~Normalisation \hfill \framebox{$M : A \norm N$}
\[
  \begin{array}{rcl}
    (\lambda x. N) \app M : A & \norm & N [ M / x ]  \\
    \Trecord{\seq{\ell = M}}.\ell : A & \norm & M_\ell  \\
    \ifelse{\vtrue}{L}{M} : A & \norm & L \\
    \ifelse{\vfalse}{L}{M} : A & \norm & M \\
    (\ifelse{L}{M}{M'}) \app N : A & \norm & \ifelse{L}{M \app N}{M' \app N} \\
      \ifelse{L}{M}{M'} : (\seq{\ell : A}) & \norm & (\seq{\ell = N}) \\
                                         & & \text{with } N_\ell = \\
                                         & & \quad \ifelse{L}{M.\ell}{M'.\ell} \\
                                         & & \text{ for each } \ell \in \seq{\ell}
  \end{array}
\]

~Evaluation \hfill \framebox{$M \evalsto V$}
\[
  \begin{array}{lrcl}
  \text{Values} & V & ::= & c \midspace \lambda x . M \midspace \Trecord{\seq{\ell = V}} \\
  \end{array}
\]
\begin{mathpar}
  \inferrule
  { }
  { V \evalsto V }

  \inferrule
  { L \evalsto \lambda x . N \\\\  M \evalsto V \\ N [ V / x ] \evalsto W }
  { L \app M \evalsto W }

  \inferrule
  { (M_i \evalsto V_i)_i }
  { \Trecord{\seq{\ell = M}} \evalsto \Trecord{\seq{\ell = V}} }

  \inferrule
  { M \evalsto \Trecord{(\ell_i = V_i)_{i \in I}} \\ j \in I }
  { M.\ell_j \evalsto V_j }

  \inferrule
  { L \evalsto \vtrue \\ M \evalsto V }
  { \ifelse{L}{M}{N} \evalsto V }

  \inferrule
  { L \evalsto \vfalse \\ N \evalsto V }
  { \ifelse{L}{M}{N} \evalsto V }

  \inferrule
  { (M_i \evalsto V_i)_i }
  { \opargs{\seq{M}} \evalsto \opargshat{\seq{V}} }
\end{mathpar}

\caption{Normalisation and Evaluation}
\label{fig:normal-forms}
\end{figure}

Given a functional language for predicates, we wish to show that predicates can
be normalised to a fragment easily translatable to SQL and usable when
typechecking lenses.
Figure~\ref{fig:normal-forms} introduces normal forms $O$ which include
variables, constants, $\lambda$-abstractions, records whose fields are all
values, record projection from a variable, conditional expressions whose
subterms are all in normal form, and operations whose arguments are all in
normal form.
Terms in \emph{predicate normal form}, ranged over by $P$, are a restriction of
terms in normal forms. Terms in predicate normal form have a
straightforward SQL equivalent, and can be used when typechecking lenses.

Normalisation rules $M \norm N$ are a subset of the rules proposed
by~\citet{Cooper09:linq-norm}: the first four rules are standard
$\beta$-reduction rules; the fifth pushes function application inside branches
of a conditional; and the sixth pushes conditional expressions inside each
component of a record. Normalisation rules can be applied anywhere in a
term, so we do not require congruence rules.

The rewrite system is strongly normalising.

\begin{proposition}[Strong normalisation]
If $\Gamma \vdash M : A$, then there are no infinite $\norm$ sequences from $M$.
\end{proposition}
\begin{proof}
A special case of the result shown by~\citet{Cooper09:linq-norm}.
\end{proof}

\emph{Static} predicates refer only to constants and properties of a given
record. Let $\norm^*$ be the
transitive, reflexive closure of the normalisation relation.
Given a variable with base record type $\trowa$, we can show that
normalisation results in a term in predicate normal form.

\begin{proposition}[Normal forms]\label{prop:normal-forms}
If $x : \trowa \vdash M : A$ and $M \norm^* N \not\norm$, then $N$ is in
normal form.
\end{proposition}
\begin{proof}
  By induction on the derivation of $x : \trowa \vdash M : A$.
  \begin{techreport}
   The details can be found in Appendix~\ref{appendix:predicates}.
 \end{techreport}
\end{proof}

As a corollary, by considering only terms with type $\Tbool$, we can show that
static predicates are in predicate normal form.

\begin{corollary}[Predicate normal form]
If $x : \trowa \vdash M : \Tbool$ and $M \norm^* N \not\norm$, then $N$ is in
predicate normal form.
\end{corollary}

Consequently, any static predicate written in our predicate language can be
normalised to predicate normal form, allowing it to be used in typechecking of
lenses and for translation into SQL. Furthermore, the normalisation procedure can
be applied to any \emph{dynamic} predicate at runtime in order to allow the same
checks to be performed dynamically.

\paragraph{Evaluation.}
Figure~\ref{fig:normal-forms} also introduces a standard big-step evaluation
relation $M \evalsto V$, which states that term $M$ evaluates to a value $V$. We
use the notation ${\scriptstyle \opargshat{\seq{V}}}$ to describe the denotation of
operation $\genop$ applied to arguments ${\scriptstyle \seq{V}}$: for example, $\hat{+}\{5, 10\} = 15$.
The semantics enjoys a standard type soundness property.

\begin{proposition}[Type Soundness]
  If $\cdot \vdash M : A$, then there exists some $V$ such that $M \evalsto V$
  and $\cdot \vdash V : A$.
\end{proposition}

\section{Typechecking Relational Lenses}\label{sec:impl:tc}

In this section, we show how na\"ive composition of lens combinators can give
rise to ill-formed lenses, and show how such ill-formed lenses can be ruled out
using static and dynamic checks. We adapt the rules proposed
by~\citet{bohannon2006relational} to the
setting of a functional programming language. We begin by discussing functional
dependencies, and then look at each lens combinator in turn.

\subsection{Functional Dependencies}

Functional dependencies are constraints restricting combinations of records. A
functional dependency ${\scriptstyle \slaba \to \slabb}$ requires that two
records with the same values for ${\scriptstyle \slaba}$ should have the same
values for ${\scriptstyle \slabb}$. We use $\setfdsa$ and $\setfdsb$ to denote
sets of functional dependencies. It is possible to derive functional
dependencies from other functional dependencies. The judgement ${\scriptstyle
  \setfdsa ~ \vDash ~
  \slaba \to \slabb}$ specifies that the functional dependency ${\scriptstyle \slaba \to \slabb}$
can be derived from the set of functional dependencies $\setfdsa$ following
\emph{Armstrong's axioms}~\cite{abiteboul1995foundations}; these (standard) derivation
rules can be found in
\begin{published}
the extended version of the paper.
\end{published}
\begin{techreport}
Appendix~\ref{appendix:supplementary}.
\end{techreport}
 The \emph{output fields} of the
functional dependencies $\setfdsa$, written $\outputs \setfdsa$, is the set of
fields constrained by $\setfdsa$ and is defined as:

\begin{defin}[Output fields] \hfill \\
  $\outputs \setfdsa = \set{\laba \in \slaba \mid \exists \slabb \in \slaba.~ \laba \notin \slabb \text{ and } \setfdsa \models \slabb \to \laba}$.
\end{defin}

\citet{bohannon2006relational} impose a special restriction on functional
dependencies called \emph{tree form}. Tree form requires that functional
dependencies form a forest, meaning that column names can be partitioned into
pairwise-disjoint sets forming a directed acyclic graph with at most one
incoming edge per node. As an example,  $\set{A \to B, A \to C, C \to D}$ is in
tree form.
It is straightforward to check whether a set of functional dependencies is in
tree form using a standard graph reachability algorithm.

Sets of functional dependencies which are semantically equivalent to a set of
functional dependencies in tree form are also considered to be in tree form. As
an example, $\set{A \to B C, C \to D}$ is not literally in tree form but is
semantically equivalent to the previous example, so can thus considered to be in
tree form.

\subsection{Lens Types}
\label{section:lens_types}

\begin{figure}
  \begin{flushleft}
  $
  \begin{array}{ll}
    \text{Table names} & \rela, \relb
  \end{array}
  $
\end{flushleft}
\[
  \begin{array}{lrcl}
    \text{Types} & A, B \! & ::= & \cdots \midspace \tablety{\trowa} \midspace \recordsetty{\trowa} \\
                 & & \midspace & \lensty{\trowa}{\preda}{\setfdsa} \\
    \text{Terms} \!\! & L, M, N \!\! & ::= & \cdots \midspace \tablee \rela
                                             \trowa \midspace
    \lens{M}{\setfdsa} \\
                 & & \mid & \lensselect{M}{P}  \\
                 & & \mid & \lensjoindl{M}{N} \\
                 & & \mid & \lensdrop{\laba'}{\slaba}{V}{M} \\
                 & & \mid & \lensget{M} \midspace \lensput{M}{N}
\end{array}
\]
\caption{Syntax of types and terms for tables and lenses}
\label{fig:lens-types-terms}
\end{figure}

Figure~\ref{fig:lens-types-terms} shows the additional types and terms for
tables and lens constructs. We let $\rela, \relb$ range over table names.
Type $\tablety[\rela]{\trowa}$ is the type of a table with table name $\rela$
containing records of type $\trowa$.
The \emph{record set type} $\recordsetty{\trowa}$ describes a set of records of
type $\trowa$.
The type of lenses, $\lensty{\trowa}{\preda}{\setfdsa}$, consists of four
components: the set of
underlying tables $\Sigma$; the base record type $\trowa$; a \emph{restriction
predicate} $\lampreda$; and a set of functional dependencies $\setfdsa$. The base
record type describes the type of rows which can be retrieved or committed to
the view, and the restriction predicate describes the subset of records on which the
lens operates.

In the remainder of the section, we describe each lens combinator and its typing
rule in turn.

\subsection{Rules}
\label{section:typing_rules}

We now introduce the rules we use to typecheck relational lenses,
adapted from the rules as defined by~\citet{bohannon2006relational} to support
nested composition and to make use of our concrete predicate syntax.
We show a formal correspondence between our typing rules and the typing rules
of~\citet{bohannon2006relational}
in~\secref{sec:impl:tc:correctness}. We first introduce some notation.

\begin{defin}[Record concatenation] \hfill
  \begin{itemize}
    \item Given records $\rowa$ = $\Trecord{\laba_1 = V_1, \ldots, \laba_m = V_m}$
      and $\rowb = \Trecord{\laba_{m + 1} = V_{m + 1}, \ldots, \ell_n = V_n}$ with disjoint field names,
      define the \emph{record concatenation} $\rowa \rowapp \rowb =
      \Trecord{\laba_1 = V_1, \ldots, \laba_n = V_n}$.
    \item Given record types $\trowa$ = $\Trecord{\laba_1 : A_1, \ldots, \laba_m : A_m}$
      and $\trowb = \Trecord{\laba_{m + 1} : A_{m + 1}, \ldots, \ell_n : A_n}$ with disjoint field names,
      define the \emph{record type concatenation} $\trowa \rowtyapp \trowb =
      \Trecord{\laba_1 : A_1, \ldots, \laba_n : A_n}$.
  \end{itemize}
\end{defin}

\paragraph{Tables.}

Links defines a primitive table expression $\tablee \rela \trowa$
which defines a handle to a table in the database. The table
expression assumes that the programmer has supplied a record type which
corresponds to the types in the underlying database schema.
\begin{mathpar}
  \inferrule[T-Table]
  { }
  {\tyto \tyctx {\tablee \rela \trowa} {\tablety \trowa}}
\end{mathpar}

\paragraph{Lens Primitives.}
The rule \ruleref{T-Lens} is used to create a relational lens from a Links
table.
A lens primitive is assigned the default predicate constraint
$\vtrue$. All columns referred to by a set of functional dependencies
$\setfdsa$, written $\nodes{\setfdsa}$, should be part of the table record type
$\trowa$.
\begin{mathpar}
  \inferrule[T-Lens]
    { \tyto{\tyctx}{\expra}{\tablety \trowa} \\
      \bigcup \nodes{\setfdsa} \subseteq \aliases{\trowa} }
    { \tyto{\tyctx}{\lens{\expra}{\setfdsa}}{\lensty[\set{\rela}]{\trowa}{\vtrue}{\setfdsa}} }
\end{mathpar}

\subsubsection{Select Lens}
The \emph{select} lens filters a view according to a given predicate.
Let us assume we have a lens $l_1$ which is the join of the two tables
\textsf{albums} and \textsf{tracks}. We might first define a lens $l_2$ to find
popular albums for which the stock is too low, by only returning the albums
where \code{quantity < rating}.

\[
  \begin{tabular}{|c|c|c|c|c|}\hline
    \texttt{track} & \texttt{year} & \texttt{rating} & \texttt{album} &
                                                                        \texttt{quantity} \\ \hline
    Lullaby & 1989 & 3 & Galore & 1 \\ \hline
    Lovesong & 1989 & 5 & Galore  & 1 \\ \hline
    Lovesong & 1989 & 5 & Paris & 4 \\ \hline
    Trust & 1992 & 4 & Wish & 4 \\ \hline
  \end{tabular}
\]
 
We might then decide to further limit this view by defining a lens $l_3$ which
only shows the tables with the album Galore.

\[
  \begin{tabular}{|c|c|c|c|c|}\hline
    \texttt{track} & \texttt{year} & \texttt{rating} & \texttt{album} &
                                                                        \texttt{quantity} \\ \hline
    Lullaby & 1989 & 3 & Galore & 1 \\ \hline
    Lovesong & 1989 & \correct{5}{4} & Galore  & 1 \\ \hline
  \end{tabular}
\]

The user then notices that the rating for \emph{Lovesong} is not correct, and
changes it from $5$ to $4$. Calling \kw{put} on $l_3$ would yield the updated
view for $l_2$:

\[
  \begin{tabular}{|c|c|c|c|c|}\hline
    \texttt{track} & \texttt{year} & \texttt{rating} & \texttt{album} &
                                                                        \texttt{quantity} \\ \hline
    Lullaby & 1989 & 3 & Galore & 1 \\ \hline
    Lovesong & 1989 & \correct{5}{4} & Galore  & 1 \\ \hline
    Lovesong & 1989 & \correct{5}{4} & Paris & 4 \\ \hline
    Trust & 1992 & 4 & Wish & 4 \\ \hline
  \end{tabular}
\]
 
Since the rating of the track Lovesong is $4$ and not lower than the quantity of
the album Paris, the updated view for $l_2$ violates the predicate requirement
\code{quantity < rating}.

To prevent such an invalid combination of lenses, the select lens needs to
ensure that the underlying lens has no predicate constraints on any fields which
may be changed by functional dependencies. The set of fields which can be
changed by functional dependencies $\setfdsa$ is $\outputs{\setfdsa}$. A
predicate $\preda$ ignores the set ${\scriptstyle \slaba}$ if the result of evaluating the
predicate $\preda$ with respect to a row in the database is not affected by
changing any fields in ${\scriptstyle \slaba}$.
\begin{defin}[Predicate Ignores]
  We say $\preda$ \emph{ignores} $\slaba$ if there exists an $\trowa$ such that
  $\slaba$ is disjoint from $\domain \trowa$ and $\tyto {x : \trowa} \preda
  \Tbool$.
\end{defin}

The \ruleref{T-Select} rule also needs to ensure that the resulting lens only
accepts records that satisfy the given predicate $\lampredb$ as well as any existing
constraints $\lampreda$ that already apply to the underlying lens. The resulting
lens's
constraint predicate can thus be defined as $\lambda x . \preda \wedge \predb$. The full select lens
typing rule can be defined as:
\begin{mathpar}
  \inferrule{\tyto \tyctx \expra {\lensty \trowa \preda \setfdsa} \\ \tyto
    {x : \trowa} \predb \Tbool \\\\
    \setfdsa \text{ is in tree form} \\
    \preda \text{ ignores }
    \outputs \setfdsa }{\tyto \tyctx {\lensselect \expra \predb}{\lensty \trowa
      {\preda \wedge
        \predb} \setfdsa}}
\end{mathpar}

\subsubsection{Join Lens}
The \emph{join} lens joins two underlying views.
A join lens has limitations on the functional dependencies of the underlying
tables. Let us assume that there is another table \texttt{reviews} which
contains album reviews by users. The table has the functional dependency
\code{user album -> review}\footnote{This example does not satisfy
  functional dependency tree form. If it instead only had the functional
  dependencies \lstinline+user -> review+, the same problem would occur.}.

\[
  \begin{tabular}{|c|c|c|}\hline
    \texttt{user} & \texttt{review} & \texttt{album} \\ \hline
    musicfan & 4 & Galore \\ \hline
    90sclassics & 5 & Galore \\ \hline
    thecure & 5 & Paris \\ \hline
  \end{tabular}
\]
 
The \texttt{reviews} table is joined with the \texttt{tracks} table to produce
the lens $l_1$. Suppose the user tries to delete the first
``90sclassics'' record:

\[
  \begin{tabular}{|c|c|c|c|c|c|}\hline
    \texttt{user} & \texttt{review} & \texttt{track} & \texttt{year} & \texttt{rating} & \texttt{album} \\ \hline
    musicfan & 4 & Lullaby & 1989 & 3 & Galore \\ \hline
    musicfan & 4 & Lovesong & 1989 & 5 & Galore \\ \hline
    \strikeStart{90sclassics} & 5 & Lullaby & 1989 & 3 & \strikeEnd{Galore} \\ \hline
    90sclassics & 5 & Lovesong & 1989 & 5 & Galore \\ \hline
    thecure & 5 & Lovesong & 1989 & 5 & Paris \\ \hline
  \end{tabular}
\]

In this case, there is no way to define a correct behaviour for \emph{put}. If the user's review is
deleted then the other entry by the same user would also be removed from the
joined table. If the track is deleted, then the entry from the other user for
the same track would also be removed.

The issue is resolved by requiring that one of the tables is completely
determined by the join key. The added functional dependency restriction ensures
that each entry in the resulting view is associated with exactly one entry in
the left table.
\jrc{I think this means: one table's fields are completely determined
  by the functional dependencies of the other table}
\rh{Is this more clear now?}
In this case, if the reviews table contained a single review per track, it would
be possible to delete any individual record by only deleting the entry in the
reviews table. In practice we need to show that we can derive the functional
dependency ${\scriptstyle \slaba \cap \slabb \to \slabb}$, where ${\scriptstyle
  \slaba \cap \slabb}$ are the join columns and ${\scriptstyle \slabb}$ is the
set of columns of the right table. We can check if this functional dependency
can be derived by calculating the transitive closure of ${\scriptstyle \slaba
  \cap \slabb}$ and then checking if ${\scriptstyle \slabb}$ is a subset.

Join lenses come in different variants with varying deletion behaviours: a
variant that always deletes the entry from the left table, a variant that tries
to delete from the right table and otherwise deletes from the left table, and a
variant that deletes the entries from both tables if possible.
The type checking for each variant is similar, so we only discuss the
delete left lens. The rule \ruleref{T-Join-Left} requires us to also show that
$\preda$ ignores $\outputs{\setfdsa}$ and $\ignoresoutputs{\predb}{\setfdsb}$.
The resulting lens should have the predicate $\preda \wedge \predb$ since the
record constraints of both input lenses apply to the output lens.

{\small
\begin{mathpar}
  \inferrule
    [T-Join-Left]
    { \tyto{\tyctx}{\expra}{\lensty \trowa \preda \setfdsa} \\
      \tyto{\tyctx}{\exprb}{\lensty[\schemab] \trowb \predb \setfdsb} \\\\
    \setfdsb \vDash \aliases \trowa \cap \aliases \trowb \to \aliases \trowb \\
    \setfdsa \text{ is in tree form} \\
    \setfdsb \text{ is in tree form} \\\\
    \ignoresoutputs{\preda}{\setfdsa} \\
    \ignoresoutputs{\predb}{\setfdsb} \\
    \schemaa \cap \schemab = \emptyset }
  { \tyto{\tyctx}{\lensjoindl{\expra}{\exprb}}
    {\lensty[\schemaa \cup \schemab] {\trowa \rowtyapp \trowb}
    {\!\preda \wedge \predb} {\setfdsa \cup \setfdsb}}}
\end{mathpar}
}

\subsubsection{Drop Lens}
The \emph{drop} lens allows a more fine-grained notion of relational projection,
allowing us to remove a column from a view. Note that this is not to be confused
with the SQL \lstinline+DROP+ statement, which deletes a table. Let us assume
we define the lens $l_1$ as a select lens with predicate \code{year > 1990} $\vee$ \code{rating > 4}.

\[
  \begin{tabular}{|c|c|c|c|}\hline
    \texttt{track} & \texttt{year} & \texttt{rating} & \texttt{album} \\ \hline
    Lovesong & 1989 & 5 & Galore  \\ \hline
    Lovesong & 1989 & 5 & Paris \\ \hline
    Trust & 1992 & 4 & Wish \\ \hline
  \end{tabular}
\]

We can then define the lens $l_2$ as $l_1$, but dropping column \code{year}
determined by \code{track} to yield the table:

\[
  \begin{tabular}{|c|c|c|}\hline
    \texttt{track} & \texttt{rating} & \texttt{album} \\ \hline
    Lovesong & \correct{5}{3} & Galore \\ \hline
    Lovesong & \correct{5}{3} & Paris \\ \hline
    Trust & 4 & Wish \\ \hline
  \end{tabular}
\]

What would the new predicate constraint be? It cannot reference the field
\code{year}, since it does not exist anymore. If it were \code{rating > 4} then
the last record would be a violation in the output view. If the predicate were
\code{true} it would violate \ruleref{PutGet}: Changing the \code{rating} from
$5$ to $3$ for the track \emph{Lovesong}, would cause it to no longer satisfy
the parent lens' predicate since it is from year $1989$ and the rating is only
$3$.

The underlying issue is the dependency between the dropped field \code{year} and
the field \code{rating}. It is not possible to define a predicate $\preda$ which
specifies if any \code{rating} value is valid independently of the drop column
\code{year}. Without being able to construct such a $\preda$, a lens cannot be well-typed.

\newcommand{\pdrop}{\preda_{\laba'}}
\newcommand{\prest}{\preda_{\slaba}}

\paragraph{Lossless Join Decomposition}
The typing rule for the drop lens requires some finer-grained checks on
predicates. We begin with some preliminary definitions.

\sjf{I would like to use standard values $V, W$ instead of lowercase
metavariables $\rowa$ for records. I'll come back to this.}

\begin{defin}[Predicate satisfaction]
We say that a record $\rowa$ \emph{satisfies} predicate $\lampreda$,
  written $\sat \rowa$, if $\subst \preda \rowa \Downarrow \vtrue$.
\end{defin}

\jrc{erdundant}

\begin{defin}[Record type inhabitants]
  We define the \emph{inhabitants} of a record type $\trowa$, written $\inh
  \trowa$, as:
  \[
    \set{ \rowa \mid \cdot \vdash \rowa : \trowa }
  \]
\end{defin}

We define $\Pset{\preda}{\trowa}$ as the equivalent set of all records of type $\trowa$ satisfying a
predicate $\preda$. The definition of $\Pset{\preda}{\trowa}$ is used to show that our
implementation is sound.

\begin{defin}[Predicate sets]
  We define the \emph{set representation of
  predicate} $\lampreda$ \emph{over} $\trowa$, written $\Pset{\preda}{\trowa}$, as:
  \[
    \set{\rowa \in \inh{\trowa} \midspace \sat{\rowa} }
  \]
\end{defin}

It is often helpful to consider only a subset of fields in a record.

\begin{defin}[Record restriction]
  Given a record $\rowa = \Trecord{\ell_1 = V_1, \ldots, \ell_m = V_m, \ldots,
  \ell_n = V_n}$, we define the \emph{record restriction of } $r$ \emph{to}
  $\ell_1, \ldots, \ell_m$, written $\rowa[\ell_1, \ldots, \ell_m]$, as
$    \Trecord{\ell_1 = V_1, \ldots, \ell_m = V_m}$.
 \end{defin}

Let $\setpreda, \setpredb$ range over homogeneous sets of records, such as the
set representation of predicates. It is also convenient to be able to consider a
set where each constituent record is restricted to a given set of fields.

\begin{defin}[Predicate set restriction]
  We define the \emph{restriction of set } $\setpreda$ \emph{to} $\seq{\ell}$,
  written $\setpreda[\seq{\ell}]$, as $\set{\rowa[\slaba] \mid \rowa \in \Pi}$.
\end{defin}

It is also useful to be able to consider the natural join of two sets of
records.

\begin{defin}[Set join]
  Suppose $\trowa = \trowa_1 \rowtyapp \trowa_2$, and suppose $\setpreda$
  contains records of type $\trowa_1$ and $\setpredb$ contains records of type
  $\trowa_2$.

  We define the \emph{set join} of $\setpreda$ and $\setpredb$, written
  $\setpreda \Join \setpredb$, as:
  \[
    \set{ \rowa \in \inh{\trowa} \midspace \rowa[\domain
      {\trowa_1}] \in \setpreda \wedge \rowa[\domain {\trowa_2}] \in \setpredb}\]
\end{defin}

To check the safety of a drop lens, we need to show that the predicate
does not impose any dependency between the value of the dropped field and any other
field.  We formalise this constraint by defining the notion of a \emph{lossless join decomposition} (LJD).

\begin{defin}[Lossless join decomposition]
  \label{def:ljd-sound}
  A \emph{lossless join decomposition} of two record types $\trowa_1$ and
  $\trowa_2$
  with respect to a predicate $P$ of type $\tyto {x : \trowa_1 \rowtyapp
  \trowa_2}
  \preda \Tbool$, written $\ljd{\lampreda}{\trowa_1}{\trowa_2}$, means that
  for all $\rowa_1, \rowa_2 \in \inh{\trowa_1}$ and $\rowb_1, \rowb_2 \in
  \inh{\trowa_2}$,
  it is the case that:
  \[
    \sat {\rowa_1 \rowapp \rowb_1} \wedge \sat {\rowa_2 \rowapp \rowb_2}
    \implies \sat {\rowa_1 \rowapp \rowb_2}
  \]
\end{defin}

Given $\trowa, \trowa_1, \trowa_2$ such that $\trowa = \trowa_1 \rowtyapp
\trowa_2$,
our definition of lossless join decomposition suffices to show that
$\Pset{\preda}{R}$ can be expressed as the natural join of $\Pset{\preda}{R}$
restricted to the fields of $\trowa_1$, with $\Pset{\preda}{R}$ restricted to the
fields of $\trowa_2$.

\begin{lemma}[Predicate set decomposition]
  \label{lem:ljd-sound}
  Suppose $\trowa = \trowa_1 \rowtyapp \trowa_2$ and $x : \trowa \vdash \preda :
  \Tbool$. If
  $\ljd{\lampreda}{\trowa_1}{\trowa_2}$, then $\Pset{\preda}{\trowa} =
  \Pset{\preda}{\trowa}[\domain
  {\trowa_1}] \Join \Pset{\preda}{\trowa}[\domain {\trowa_2}]$.
\end{lemma}

\begin{proof}
  Follows from the definitions of $\ljd \lampreda {\trowa_1} {\trowa_2}$, $\cdot[\cdot]$ and
  $\cdot \Join \cdot$.
  \begin{techreport}
  For further details see Appendix
  \ref{section:appendix:drop-lens}.
  \end{techreport}
\end{proof}

Showing $\ljd{\lampreda}{\trowa}{\trowb}$ is NP-hard and could be undecidable,
depending on the atomic formulae available in the predicates.
Since a predicate that satisfies $\ljd{\lampreda}{\trowa}{\trowb}$ can be rewritten as a
conjunction of predicates which depend only on either $\trowa$ or $\trowb$, we
can, however, define a sound but incomplete syntactic approximation
$\ljdi{\lampreda}{\trowa}{\trowb}$.

\begin{mathpar}
  \inferrule
    [LJD$^\dagger$-1]
    { \tyto {x : \trowa} \preda \Tbool }
    {\ljdi{\lampreda}{\trowa}{\trowb}}

  \inferrule
    [LJD$^\dagger$-2]
    { \tyto {x : \trowb} \preda \Tbool }
    {\ljdi{\lampreda}{\trowa}{\trowb}}

  \inferrule
    [LJD$^\dagger$-And]
    {\ljdi{\lampreda}{\trowa}{\trowb} \\\\ \ljdi{\lampredb}{\trowa}{\trowb} }
    {\ljdi {\lambda x .~ \preda \wedge \predb}{\trowa}{\trowb}}
\end{mathpar}

\begin{lemma}[Soundness of LJD$^\dagger$]
  Given a predicate $\lampreda$ and record types $\trowa, \trowb$, it follows that
  $\ljdi \lampreda \trowa \trowb$ implies $\ljd \lampreda \trowa \trowb$.
  \label{lem:ljdi-sound}
\end{lemma}

\begin{proof}
  By induction on the derivation of $\ljdi{\lampreda}{\trowa}{\trowb}$.

  \sjf{Is there a full proof of this?}
\end{proof}

Updates to the view will use the default value $\vala$ in place of the given
column. Therefore, in addition to showing that the predicate does not impose any
dependency between the value of the dropped field and the other fields, we must
show that the default value $\vala$ of the dropped column does not violate the
predicate.
Given the set representation of a predicate $\Pset{\preda}{\trowa}$, we must show
that $\set{\laba' = \vala} \in \Pset{\preda}{\trowa}[\laba']$.

We define a property $\defv{\lampreda}{\rowa}$ and show that it is sound with
respect to the set semantics.

\begin{defin}
  Given a predicate $\lampreda$ and record types $\trowa$ and $\trowb$
  such that $\ljd{\lampreda}{\trowa} {\trowb}$ and
  $\rowa \in \inh {\trowb}$, we write
  $\defv{\lampreda}{\rowa}$ when $\Pset{\preda}{\trowa \rowtyapp \trowb}$ is not empty and there
  exists an
  $\rowb \in \inh{\trowa}$ such that $\sat {\rowa \rowapp \rowb}$.
\end{defin}

\begin{lemma}
  \label{lem:defv-sound}
  Suppose $\trowa = \trowa_1 \rowtyapp \trowa_2$, $\rowa \in \inh{\trowa_2}$, and $\ljd
  {\lampreda}{\trowa_1} {\trowa_2}$. Then $\defv[\trowa_1, \trowa_2]{\lampreda}{\rowa}$ implies $\rowa \in
  \Pset{\preda}{\trowa}[\domain{\trowa_2}]$.
\end{lemma}

\begin{proof}
  By expansion of the definitions of $\defv[\trowa_1, \trowa_2] \lampreda \rowa$ and of $\cdot \rowapp \cdot$ and $\cdot \in \cdot$.
  \begin{techreport}
  For details see Appendix \ref{section:appendix:drop-lens}.
\end{techreport}
\end{proof}

As with the definition of $\ljd{\lampreda}{\trowa}{\trowb}$, determining if
\break $\defv{\lampreda}{\rowa}$ holds in the general case is a difficult
problem. To simplify this problem we introduce an incomplete set of inference
rules to determine $\defvi{\lampreda}{\rowa}$, which covers the same set of
predicates as the $\ljdi{\lampreda}{\trowa}{\trowb}$ rule.

\begin{mathpar}
  \inferrule [DV$^\dagger$-1] {\tyto {x : \trowa} {\preda} {D} }
  {\defvi{\lampreda}{\rowa}}

  \inferrule
    [DV$^\dagger$-2]
    { \tyto {x : \trowb} {\preda} {D} \\\\ \sat \rowa }
    {\defvi{\lampreda}{\rowa}}

  \inferrule
    [DV$^\dagger$-And]
    { \defvi{\lampreda}{\rowa} \\\\ \defvi{\lampredb}{\rowa}}
    {\defvi{\lambda x .~ \preda \wedge \predb}{\rowa}}
\end{mathpar}

\begin{lemma}\label{lem:defvi-sound}
  Given a predicate $\lampreda$ such that $\Pset \preda {\trowa \rowtyapp
  \trowb}$ is not empty and record $\rowa$ such that $\cdot \vdash \rowa :
  \trowa$, it follows that $\defvi{\lampreda}{\rowa}$ implies
  $\defv{\lampreda}{\rowa}$.
\end{lemma}

\begin{proof}
  By induction on the derivation of $\defvi{\lampreda}{\rowa}$.
  \begin{techreport}
    For further
  details see Appendix \ref{lem:defvi-sound-proof}.
\end{techreport}
\end{proof}

Note that the soundness proof for $\defvi \lampreda \rowa$ requires that $\Pset
\preda {\trowa \rowtyapp \trowb}$ is not empty. This is problematic in theory,
because it requires us to show that the predicate is satisfiable. According to
\citet{bohannon2006relational}, a drop lens on a lens with predicate that is
$\vfalse$ does not typecheck. In practice however, this lens is well behaved as
it returns an empty view and only takes an empty view. The lens would therefore
be useless, but not incorrect.

With the preliminaries in place, we can present the typing rule for the drop lens.
The term $\lensdrop{\laba'}{\slaba}{\vala}{\expra}$ constructs a lens
which removes column $\laba'$ from view $\expra$, given that
the functional dependencies of the view ensure that $\laba'$ is determined by
the columns ${\scriptstyle \slaba}$.
The typing rule is as follows:

\begin{small}
  \begin{mathpar}
    \inferrule[T-Drop]
    { \setfdsa \equiv \setfdsb \cup \set{\slaba \to \laba'} \\
      \tyto{\tyctx}{\expra}{\lensty{\trowa \rowtyapp \Trecord{\oftype{\laba'}{\typpa}}}
        \preda \setfdsa} \\\\
      \slaba \subseteq \aliases \trowa \\ \tyto \tyctx \vala A\\
      \ljd{\lampreda}{\trowa}{\Trecord{\laba' : \typpa}} \\
      \defv[\trowa,\Trecord{\laba' : A}]{\lampreda}{\Trecord {\laba' = \vala}}
      \\
      \preda' = \subst[x.\laba']{\preda}{\vala}
    }
      { \tyctx \vdash \lensdrop{\laba'}{\slaba}{\vala}{\expra} {:}
        {\lensty{\trowa}{\lambda x . \preda'}{\setfdsb}} }
  \end{mathpar}
\end{small}

The clause ${\scriptstyle \setfdsa ~ \equiv ~ \setfdsb \cup \set{\slaba \to \laba'}}$ checks that
the functional dependencies of the underlying lens $M$ imply that ${\scriptstyle
\slaba}$ do
indeed determine $\laba'$; that ${\scriptstyle \slaba}$ are contained in the domain of the
record type $\trowa$ of underlying lens $M$; that $V$ has the same type as the
dropped field; that $R$ and $(\ell' : A)$ define a lossless join
decomposition with respect to the lens predicate; and finally that $V$ is a
suitable default value with respect to the predicate.

The resulting type
$\lensty{\trowa}{\preda[\vala / x.\laba']}{\setfdsb}$ contains the
updated record type without the dropped column, and the updated predicate with
the default variable in place of all references to the dropped column.

\paragraph{Lens Get}

Finally we define typing rules for making use of relational lenses. Since Links
is not dependently typed, we discard the constraints which apply to the view, and
specify that calling \kw{get} returns a set of records which all have the type
$\trowa$.
\begin{mathpar}
  \inferrule[T-Get] { \tyto{\tyctx}{\expra}{\lensty \trowa \preda \setfdsa}}{
  \tyto{\tyctx}{\lensget{\expra}}{\kw{record set of } \trowa}}
\end{mathpar}

\paragraph{Lens Put}

Just as with \ruleref{T-Get}, we have no way of statically ensuring that the
input satisfies $\preda$ and $\setfdsa$, so we only statically check that the
updated view is a set of records matching type $\trowa$, deferring the checks to
ensure that the set of records satisfies $\setfdsa$ and $\preda$ to runtime.

To ensure that the constraint $\preda$ applies to each record $\rowa$ in a view,
runtime checks ensure that $\sat \rowa$. Functional dependency constraints can
be checked by projecting the set of records down to each functional dependency
and determining if any two records violate a functional dependency.
\begin{mathpar}
  \inferrule[T-Put] { \tyto \trowa \expra {\lensty{\trowa}{\lambda x . \preda}{\setfdsa}} \\
    \tyto{\tyctx}{\exprb}{\kw{record set of } \trowa}}{
    \tyto{\tyctx}{\lensput{\expra}{\exprb}}{\Tunit} }
\end{mathpar}

\subsection{Correctness}\label{sec:impl:tc:correctness}
\citet{bohannon2006relational} prove that lenses satisfying
correctness conditions are well-behaved (i.e., satisfy \textsc{GetPut} and
\textsc{PutGet}, and therefore safely compose). Their typing rules are not in a
form amenable to implementation, since predicates are defined as abstract sets;
lenses are composed using a sequential composition operator rather than allowing
arbitrarily-nested lenses as one would in a functional language; and there is no
distinction between a relation and a lens on a relation.

Nevertheless, we must show that our typing rules also guarantee
well-behavedness. Our approach is to define a type-preserving translation from our functional-style
lenses into the sequential-style lenses defined
by~\citet{bohannon2006relational}.

\begin{figure}
  \totheleft{Syntax of sequential lenses}
    \[
      \begin{array}{ll@{\qquad}ll}
        \text{Set predicates} & \setpreda, \setpredb &
        \text{Schemas} & \Sigma, \Delta \\
        \text{Sequential lenses} & \bpvlensa, \bpvlensb ::=
      \end{array}
      \]
      \[
      \begin{array}{lrcl}
          &   &           & \lensid \midspace \bpvcompose{\bpvlensa}{\bpvlensb} \vphantom{\slabb} \\
          &   & \midspace & \bpvselect{\rela}{\setpreda}{\relb} \vphantom{\slabb} \\
          &   & \midspace & \bpvjoin{\rela_1}{\rela_2}{\relb} \vphantom{\slabb} \\
          &   & \midspace & \bpvdrop{\laba}{\slabb}{V}{\rela}{\relb}
      \end{array}
    \]
    \vspace{1em}

    ~Flattening translation \hfill \framebox{$\ltrans{M} = \Sigma / I / \rela$}
    \[
      \begin{array}{l}
    \ltrans{\lens{\rela}{\setfdsa}}  =  \set{\rela} / \lensid / \rela \\
    \ltrans {\lensselect \expra \preda}  = \\
    \qquad \schemaa / \bpvlensa; \bpvselect{\rela}{\Pset{\preda}{\domain{\rela}}}{\relb} / \relb \\
    \quad \text{ where } \ltrans{\expra} = \schemaa / \bpvlensa / \rela \text{ and } \guid \relb \\
    \ltrans {\lensjoindl \expra \exprb}  = \\
    \qquad \schemaa \uplus \schemab / \bpvlensa; \bpvlensb;
    \bpvjoin{\rela_1}{\rela_2}{\relb} / \relb \\
    \quad \text{ where } \ltrans{\expra} = \schemaa / \bpvlensa / \rela_1, \ltrans{\exprb} = \schemab / \bpvlensb / \rela_2 \text{ and } \guid \relb \\
    \ltrans {\lensdrop{\laba}{\slabb}{\vala}{\expra}}  = \\
    \qquad \schemaa / \bpvlensa; \bpvdrop{\laba}{\slabb}{\vala}{\rela}{\relb} /
    \relb \\
    \quad \text{ where } \ltrans{\expra} = \schemaa / \bpvlensa / \rela
  \end{array}
  \]%
    \caption{Sequential-style lenses~\cite{bohannon2006relational} and
    flattening}
    \label{fig:bpv-syntax}
\end{figure}
Figure~\ref{fig:bpv-syntax} shows the grammar of sequential-style lenses.
We let $\setpreda$ range over set-style
predicates; $\rela, \relb$ range over relation names; $\schemaa, \schemab$ range over
schemas (i.e., sets of relation names); and $\bpvlensa, \bpvlensb$ range over
sequential-style lenses.
The \emph{sort} of a relation $\rela$, written $\rlsort{\rela} = ({\scriptstyle \slaba},
\setpreda, \setfdsa)$, is a 3-tuple of
the set of fields ${\scriptstyle \slaba}$ in $\rela$; a set predicate $\setpreda$,
and the set of functional dependencies $\setfdsa$.
If $\rlsort{\rela} = ({\scriptstyle \slaba}, \setpreda, \setfdsa)$, then $\domain{\rela} =
{\scriptstyle \slaba}$.

Sequential-style lenses map source schemas to view schemas.  The $\lensid$ lens
defines the \emph{identity} lens, mapping a schema to itself, and
$\bpvcompose{\bpvlensa}{\bpvlensb}$ composes lenses $\bpvlensa$ and
$\bpvlensb$. The $\bpvselect{\rela}{\setpreda}{\relb}$ lens filters relation
$\rela$ using
predicate set $\setpreda$, naming the resulting relation $\relb$. The
$\bpvjoin{\rela_1}{\rela_2}{\relb}$ lens joins relations $\rela_1$ and $\rela_2$
using the delete-left strategy, naming the resulting relation $\relb$.

Finally, $\bpvdrop{\laba}{{\scriptstyle \slabb}}{V}{\rela}{\relb}$ drops attribute
$\laba$ determined by attributes ${\scriptstyle \slabb}$ with default value $V$ from relation
$\rela$, naming the resulting relation $\relb$.

Figure~\ref{fig:bpv-syntax} also shows the translation from functional lenses to
sequential-style lenses, which involves flattening functional lenses by introducing
intermediate relations with fresh table names.
The translation function $\ltrans{M} = \Sigma / I / \rela$ states that
functional lens $M$ depends on tables $\Sigma$, translates to sequential lens
$I$, and produces a view with name $\rela$.

As an example of a typing rule for sequential-style lenses, consider the typing
rule for the select lens:
\begin{mathpar}
  \inferrule[T-Select-RL]
    { \rlsort \rela = (\slaba, \setpredb, \setfdsa) \\
      \rlsort \relb = (\slaba, \setpreda \cap \setpredb, \setfdsa) \\\\
      \setfdsa \text{ is in tree form} \\ \setpredb \text{ ignores } \outputs{\setfdsa}
    }
    {\bpvselect{\rela}{\setpreda}{\relb} \in \rldisjoint{\set{\rela}}{\set{\relb}}}
\end{mathpar}

The sequential lens typing judgement has the shape $I \in \Sigma \Leftrightarrow
\Delta$, meaning that $I$ is a lens transforming the source schema $\Sigma$ into the
view schema $\Delta$.
In the case of the select lens, given a predicate set $\setpreda$, the typing
rule enforces the invariant that the source relation $\rela$ has sort
$({\scriptstyle \slaba}, \setpredb, \setfdsa)$; that the functional dependencies
$\setfdsa$ are in tree form; that $\setpredb$ ignores the outputs of $\setfdsa$;
and assigns the view $\relb$ the sort $({\scriptstyle \slaba}, \setpreda \cap
\setpredb, \setfdsa)$.

We can now state our soundness theorem, stating that once translated, lenses
typeable in our system are typeable using the original rules proposed
by~\citet{bohannon2006relational}, and can use the incremental semantics
described by~\citet{horn2018incremental}.

\begin{theorem}[Soundness of Translation]\hfill \\
  If $\tyto \tyctx \expra {\lensty \trowa \preda \setfdsa}$ and $\ltrans
  \expra = \ltriple \lensa \relb$, then $\lensa \in \schemaa \Leftrightarrow
  \set{\relb}$ and $\bpvsort{\relb} = (\domain{\trowa}, \Pset{\preda}{\trowa}, \setfdsa)$.
\end{theorem}
\begin{proof}
  By induction on the derivation of \\
  \begin{published}
  $\Gamma \vdash \expra : {\lensty{\trowa}{\preda}{\setfdsa}}$.
  \end{published}
  \begin{techreport}
  $\Gamma \vdash \expra : {\lensty{\trowa}{\preda}{\setfdsa}}$;
  see Appendix~\ref{appendix:typechecking}.
  \end{techreport}
\end{proof}

\subsection{Typechecking Dynamic Predicates}%
If a dynamic predicate is used in any lens combinator, the same checks
are performed, but checking of predicates must be deferred to runtime.
In this case, we require the programmer to acknowledge that the lens
construction may fail at run-time. We introduce a special lens, the \kw{check}
lens, which the user must incorporate prior to using the lens in a \kw{get} or
\kw{put} operation.

\section{Case Study: Curated Scientific Databases}\label{sec:case-study}
In this section, we illustrate the use of relational lenses in the setting of a
larger Links application: part of the curation interface for a scientific
database. Scientific databases collect information about a particular
topic, and are \emph{curated} by subject matter experts who manually enter and
update entries.

The IUPHAR/BPS Guide to Pharmacology (GtoPdb)~\cite{gtopdb} is a
curated scientific database which collects information on pharmacological
targets, such as receptors and enzymes, and \emph{ligands} such as
pharmaceuticals which act upon targets. GtoPdb consists of a PostgreSQL
database, a Java/JSP web application frontend to the database, and a Java GUI
application used for data curation.

In parallel work~\cite{FowlerHSC20:links-gtopdb}, we have implemented a
workalike frontend application in Links, using the Links LINQ functionality. In
this section, we demonstrate how we are beginning to use relational lenses for
the curation interface, and show how relational lenses are useful in tandem with
the Model-View-Update (MVU) paradigm pioneered by the Elm programming
language~\cite{elm-lang}.

\subsection{Disease Curation Interface}

\begin{figure}
\begin{subfigure}{\columnwidth}
  \centering
  \includegraphics[width=0.8\textwidth]{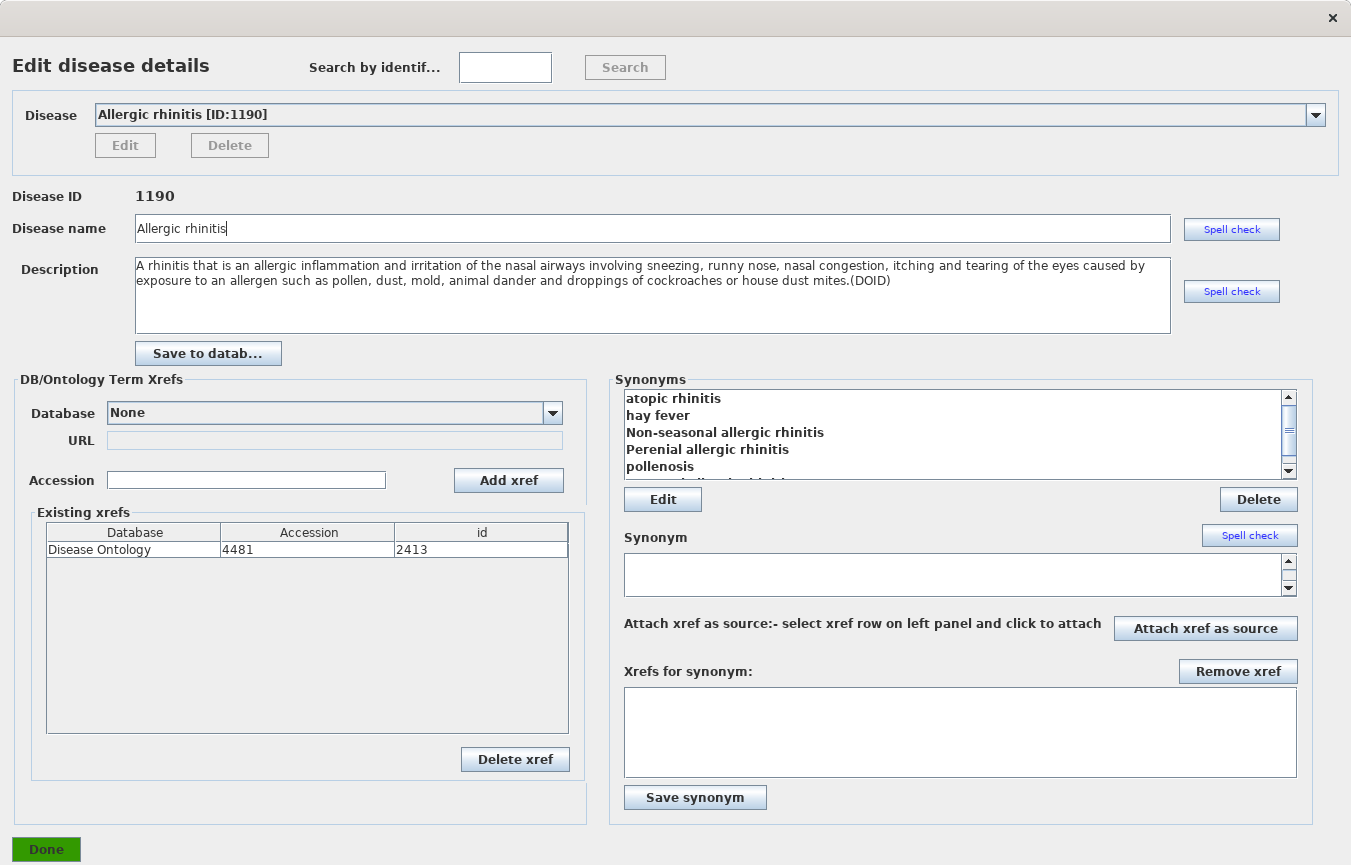}
\caption{Java curation interface}
\label{fig:java-curation}
\end{subfigure}

\begin{subfigure}{\columnwidth}
  \centering
  \includegraphics[width=0.8\textwidth]{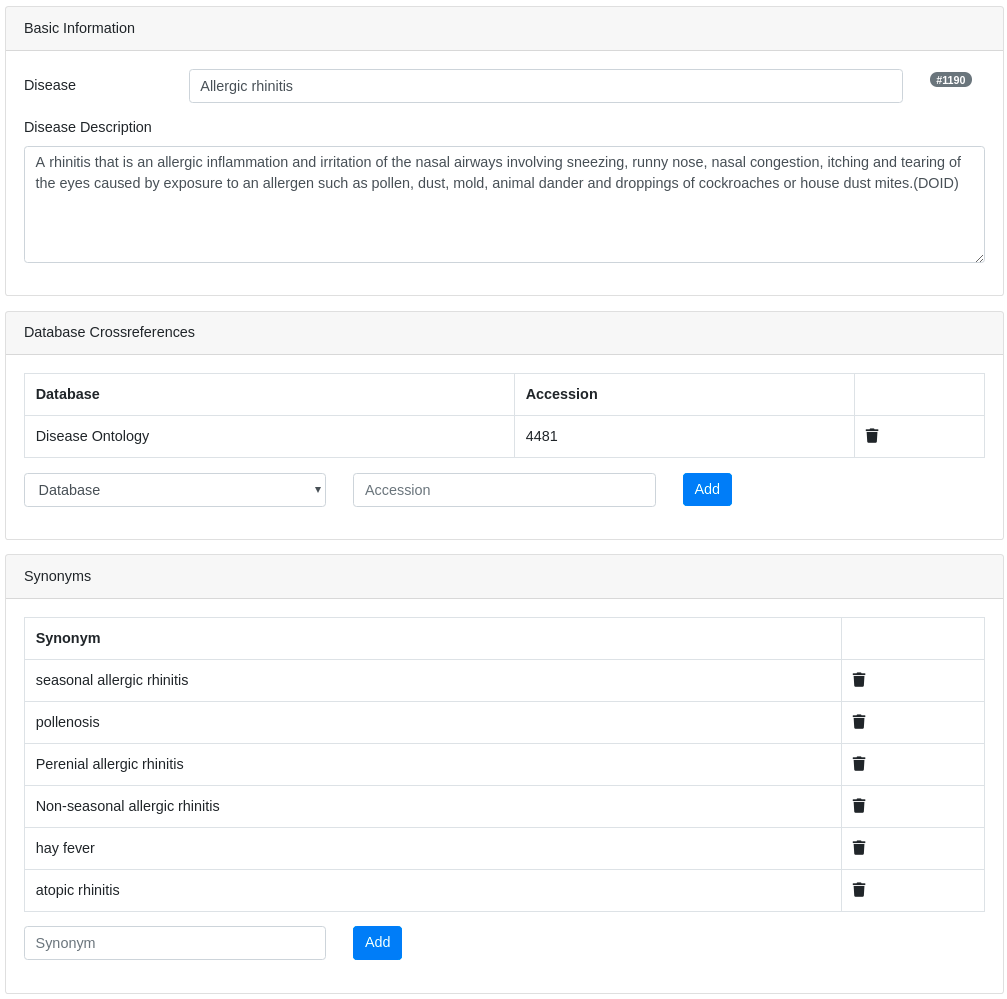}
\caption{Links reimplementation}
\label{fig:links-curation}
\end{subfigure}

\caption{Curation interfaces for Diseases}
\label{fig:disease-curation}
\end{figure}

One section of GtoPdb collects information on diseases, such as the disease name,
description, crossreferences to other databases, and relevant drugs and targets.
In this section, we describe a curation interface for diseases, where all
interaction with the database occurs using relational lenses.

Figure~\ref{fig:java-curation} shows the official Java curation interface. The
main data entries edited using the curation interface are the name and description of the
disease; the crossreferences for the disease which refer to external databases;
and the synonyms for a disease. As an example, a synonym for ``allergic
rhinitis'' is ``hayfever''.
Note that this curation interface does not edit ligand or target information;
curation of ligand-to-disease and target-to-disease links are handled by the
ligand and target curation interfaces respectively.

\subsection{Links Reimplementation}
Figure~\ref{fig:links-curation} shows the curation interface as a Links web
application.
In the original implementation of Links~\cite{cooper2006links}, requests invoked
Links as a CGI script. Modern Links web applications execute as follows:

\begin{enumerate}
  \item A Links application is executed, which registers URLs against page
    generation functions, and starts the webserver
  \item A request is made to a registered URL, and the server runs the
    corresponding page
    generation function
  \item The page generation function may spawn server processes, make
    database queries, and register processes to run on the client,
    before returning HTML to the client
  \item The client application spawns any client processes, and renders the HTML
  \item Client processes can communicate with server processes over a WebSocket
    connection.
\end{enumerate}

\subsubsection{Architecture}

The disease curation interface consists of a persistent server process, and a
client process which is spawned by the Links MVU library.

Upon page creation, the application creates lenses to the underlying tables: the
lenses retrieve data from, and propagate changes to, the database.
Since lenses only exist on the server and cannot be serialised to the client, we
spawn a process which awaits a message from the client with the updated data.

\subsubsection{Tables and Lenses.}
We begin by defining the records we need, and handles to the underlying
database and its tables.

First, we define a database handle, \lstinline+db+, to the \lstinline+gtopdb+
database.
\begin{lstlisting}
var db = database "gtopdb";
\end{lstlisting}

Next, we define type aliases for the types of records in each table. The disease
curation interface uses tables describing four entity types: disease data
(\lstinline+DiseaseData+), metadata about external databases
(\lstinline+ExternalDatabase+), links from diseases to external databases
(\lstinline+DatabaseLink+), and disease synonyms (\lstinline+Synonym+). (Note
that \lstinline+"prefix"+ appears in quotes as \lstinline+prefix+ is a Links
keyword).

\begin{lstlisting}
typename DiseaseData =
  (disease_id: Int, name: String,
   description: String, type: String);
typename ExternalDatabase =
  (database_id: Int, name: String, url: String,
   specialist: Bool, "prefix": String);
typename DatabaseLink =
  (disease_id: Int, database_id: Int, placeholder: String);
typename Synonym = (disease_id: Int, synonym: String);
\end{lstlisting}

We will need to join the \lstinline+ExternalDatabase+ and
\lstinline+DatabaseLink+ tables in order to render the database name of each
external database link. It is therefore useful to define a type synonym for the
record type resulting from the join:

\begin{lstlisting}
typename JoinedDatabaseLink =
  (disease_id: Int, database_id: Int, placeholder: String,
   name: String, url: String,
   specialist: Bool, "prefix": String);
\end{lstlisting}

Next, we can define handles to each database table. The \lstinline+with+ clause
specifies a record type denoting the column name and type of each attribute in
the table, and the \lstinline+tablekeys+
clause specifies the primary keys (i.e., sets of attributes which uniquely
identify a row in the database) for each table. We show only the definition of
\lstinline+diseaseTable+; the definitions for \lstinline+databaseTable+,
\lstinline+dbLinkTable+, and \lstinline+synonymTable+ are similar.

\begin{lstlisting}
var diseaseTable =
  table "disease" with DiseaseData
  tablekeys [["disease_id"]] from db;
\end{lstlisting}

The ID of the disease to edit (\lstinline+diseaseID+) is provided as a
\lstinline+GET+
parameter to the page, and thus we need a dynamic predicate as not all
information is known statically.
With the description of the entities and tables defined, we can describe the
relational lenses over the tables. We work in a function scope where
\lstinline+diseaseID+ has been extracted from the \lstinline+GET+ parameters.

\begin{lstlisting}
fun diseaseFilter(x) { x.disease_id == diseaseID }
# Disease lenses
var diseasesLens = lens diseaseTable default;
var diseasesLens =
  check (select from diseasesLens by diseaseFilter);(*\vspace{0.5em}*)
# Database link lenses
var dbLens = lens databaseTable
  with { database_id -> name url specialist "prefix" };
var dbLinksLens = lens dbLinkTable default;
var dbLinksLens =
  check (select from dbLinksLens by diseaseFilter);
var dbLinksJoinLens = check (
  join dbLinksLens with dbLens
    on database_id delete_left);(*\vspace{0.5em}*)
# Synonym lenses
var synonymsLens = lens synonymTable default;
var synonymsLens =
  check (select from synonymsLens by diseaseFilter);
\end{lstlisting}

We create a lens over a table using the \lstinline+lens+ keyword, writing
\lstinline+default+ when we do not need to specify functional dependencies. The
\lstinline+dbLens+ lens specifies a functional dependency from
\lstinline+database_id+ to each of the other columns, as knowledge of this
dependency is required when constructing a join lens.

We need not filter the \lstinline+databaseTable+ table since we wish to display
all external databases.
The \lstinline+diseaseLens+, \lstinline+dbLinksLens+, and
\lstinline+synonymsLens+ lenses make use of the
\lstinline+select+ lens combinator, allowing us to consider only the records
relevant to the given \lstinline+diseaseID+.  Note that each entity has a
\lstinline+disease_id+ field: as a result, we can make use of Links' row typing
system~\cite{LindleyC12:rows} to define a \emph{single}
predicate, \lstinline+diseaseFilter+, for each select lens using row polymorphism.

The \lstinline+dbLinksJoinLens+ lens joins the external database links with the
data about each external database by using the \lstinline+join+ lens combinator,
stating that if a record is deleted from the view, then it should be deleted from the
\lstinline+dbLinkTable+ rather than the \lstinline+dbLens+ table.
Joining these two tables is only possible because \lstinline+database_id+
uniquely determines each column of the \lstinline+databaseTable+ table; as the
lens uses a dynamic predicate, this property is checked at runtime.

\subsubsection{Model}

In implementing the case study, we make use of the \emph{Model-View-Update}
(MVU) paradigm, pioneered by the Elm programming language~\cite{elm-lang}. MVU is
similar to the Model-View-Controller design pattern in that it splits the state
of the system from the rendering logic. In contrast to MVC, MVU relies on explicit message
passing to update the model. The key interplay between MVU and relational lenses
is that MVU allows the model to be directly modified in memory, and relational
lenses allow the changes in the model to be directly propagated to the database
\emph{without} writing any marshalling or query construction code.

\begin{lstlisting}
typename DiseaseInfo =
  (diseaseData: DiseaseData, databases: [ExternalDatabase],
   dbLinks: [JoinedDatabaseLink], synonyms: [Synonym]);

typename Model =
  Maybe(
    (diseaseInfo: DiseaseInfo, selectedDatabaseID: Int,
     accessionID: String, newSynonym:String,
     submitDisease: (DiseaseInfo) {}~> ()));
\end{lstlisting}%
The model (\lstinline+Model+) contains all definitions retrieved from the database
(\lstinline+DiseaseInfo+), as well as the current value of the various form
components for adding database links (\lstinline+selectedDatabaseID+ and
\lstinline+accessionID+) and synonyms (\lstinline+newSynonym+). Finally, the
model contains a function \lstinline+submitDisease+ which commits the
information to the database. Note that the \lstinline+{}~>+ function arrow
denotes a function which cannot be run on the database, and does not perform any
effects.
The \lstinline+Model+ type is wrapped in a \lstinline+Maybe+ constructor to
handle the case where the application tries to curate a nonexistent disease.

\paragraph{Initial model.}
To construct the initial model, we fetch the data
from each lens using the \lstinline+get+ primitive. We include type annotations
for clarity, but they are not required.

\begin{lstlisting}
var (diseases: [DiseaseData])       = get diseasesLens;
var (dbs: [ExternalDatabase])       = get dbLens;
var (dbLinks: [JoinedDatabaseLink]) = get dbLinksJoinLens;
var (synonyms: [Synonym])           = get synonymsLens;
\end{lstlisting}

Next, we spawn a server process which awaits the submission of an updated
\lstinline+DiseaseInfo+ record. The \lstinline+Submit+ message contains the
updated record along with a client process ID \lstinline+notifyPid+ which is
notified when the query is complete.

The \lstinline+submitDisease+ function takes an updated \lstinline+DiseaseInfo+
process ID and sends a \lstinline+Submit+ message to the server. The
\lstinline+spawnWait+ keyword spawns a process, waits for it to complete, and
returns the retrieved value. In our case, we use \lstinline+spawnWait+ to only
navigate away from the page once the query has completed.

\begin{lstlisting}
var pid = spawn {
  receive {
    case Submit(diseaseInfo, notifyPid) ->
      put diseasesLens with [diseaseInfo.diseaseData];
      put dbLinksJoinLens with diseaseInfo.dbLinks;
      put synonymsLens with diseaseInfo.synonyms;
      notifyPid ! Done
  }
};

sig submitDisease : (DiseaseInfo) {}~> ()
fun submitDisease(diseaseInfo) {
  spawnWait {
    pid ! Submit(diseaseInfo, self());
    receive { case Done -> () }
  };
  redirect("/editDiseases")
}
\end{lstlisting}

Given the above, we can construct the initial model. Recall that the result of
\lstinline+get diseasesLens+ is a \emph{list} of \lstinline+DiseaseInfo+
records. As \lstinline+disease_id+ is the primary key for the
\lstinline+disease+ table, we know that the result set must be either empty or a
singleton list. Finally, we can initialise the model with the data retrieved
from the database along with the \lstinline+submitDisease+ function and default
values for the form elements.

\begin{lstlisting}
var (initialModel: Model) = {
  switch(diseases) {
    case [] -> Nothing
    case d :: _ ->
      var diseaseInfo =
        (diseaseData = d, databases = dbs,
         dbLinks = dbLinks, synonyms = synonyms);
      Just((diseaseInfo = diseaseInfo,
            accessionID = "", newSynonym = "",
            selectedDatabaseID = hd(dbs).database_id,
            submitDisease = submitDisease))
  }
};
\end{lstlisting}

The model is rendered to the page using a \lstinline+view+ function which takes
a model and produces some HTML to display. Interaction with the page produces
\emph{messages} which cause changes to the model. Finally, submission of the
form causes the \lstinline+submitDisease+ function to be executed, which in turn
sends a \lstinline+Submit+ message to the server to propagate the changes to the
database using the lenses.

\subsection{Discussion}
In this section, we have described part of the curation interface for a
scientific database. Our application is a tierless web application
with the client written using the Model-View-Update architecture.

Relational lenses allow seamless integration between all three layers of the
application. Lenses with dynamic predicates allow us to retrieve the
relevant data from the database; the data is used as part of a model which
is changed directly as a result of interaction with the web page; and the
updated data entries are committed directly to the database. At no
point does a user need to write a query: every interaction with the database
uses only lens primitives.

The primary limitation of the implementation at present is that it does not
currently support auto-incrementing primary keys, which are commonly used in
relational databases.
\section{Related Work}\label{sec:related}
\paragraph{Edit Lenses.}
Edit lenses are a form of bidirectional transformation where, rather than
translating directly between one data structure and another, the changes to a
data structure are tracked and then translated into changes to the other data
structure \cite{hofmann2012edit}. They can be particularly useful in the case of
\emph{symmetric lenses} in situations where neither of the data structures
contain all of the data, and thus none of the sources can be considered the
`source' \cite{hofmann2011symmetric}. Changes could be described by \emph{insert},
\emph{update} and \emph{delete} commands, and will usually result in similar
insert, update or deletion commands for the other data structure.

Relational lenses are generally not considered edit lenses, as they directly
translate the entire view to an updated source when performing get.
\emph{Incremental Relational Lenses} on the other hand take the updated view and
compute a delta which is then translated into a delta to the source tables
\cite{horn2018incremental}.

The language integration aspect of relational lenses is not dependent on the
semantics used to perform relational updates. Instead it only relies on all of
the relational lens typing rules in \secref{section:typing_rules} to be
satisfied; in this case, both the incremental and the
non-incremental relational put semantics are guaranteed to be well-behaved.

\paragraph{Put-based Lenses.}
Bidirectional lenses are often defined in a form that corresponds to
the forward (get) direction and the reverse direction. A common issue with this
approach is that a get function might correspond to several well-behaved put
functions, as illustrated by drop and join relational lenses.
As such, defining a bidirectional transformation by
only specifying the forward direction is generally not sufficient. An
alternative approach recently used is to rather have the programmer instead only
specify the \emph{put} semantics, which then uniquely define the \emph{get}
semantics \cite{FischerHP15:putback,hu2014validity}.

A \emph{putback} approach to bidirectional transformations has been recently
proposed by Asano et al. \cite{asano2018view} for relational data. Asano
et al.\ define a language which allows the specification of update queries, for
which the forward query can automatically be derived. They support splitting
views vertically for defining behaviour specific to columns and horizontally for
behaviour specific to rows. For each of the different sections of the view they
can then define the update behaviour, which can be simple checks or actual
update semantics.
\jrc{Can we give any details of differences with their approach?  Is
  there a reason to prefer theirs or ours?  Do they have an
  implementation that works on nontrivial amounts of data?  Can they
  express all of the relational lens operations?}
\paragraph{Cross-tier web programming.}

SMLServer~\cite{ElsmanH03:smlserver} was among the first functional frameworks
to allow interaction with a relational database.  Ur/Web~\cite{Chlipala15:urweb}
is a cross-tier web programming language which supports a statically-typed SQL
DSL, along with atomic transactions and functional combinators on results.
Neither framework supports language-integrated views.

Hop.js~\cite{SerranoP16:hopjs} builds on the Hop programming
language~\cite{SerranoGL06:hop} and allows
cross-tier web programming in JavaScript.
Eliom~\cite{RadannePVB16:eliom} is a cross-tier functional
programming framework building on top of the OCaml programming language. Eliom
programs can explicitly assign locations to functions and variables.
ScalaLoci~\cite{WeisenburgerKS18:scalaloci} is a Scala framework for cross-tier
application programming.  A key concept behind ScalaLoci is that data transfer
between tiers uses the \emph{reactive programming} paradigm.
Haste.App~\cite{Ekblad16:haste} is a Haskell EDSL allowing web applications to
be written directly in Haskell.
Since these are embedded DSLs or frameworks,
it becomes possible to use the database functionality provided by other
libraries, but are not aware of any work providing relational lenses as a library
in any programming language.

Task-oriented programming (TOP)~\cite{PlasmeijerLMAK12:top} is a high-level paradigm
centred around the idea of a \emph{task}, which can be thought of as a unit of
work with an observable value. TOP is implemented in the iTask
system~\cite{PlasmeijerAK07:itasks}.
An \emph{editor} is a task which interacts with a
user. \emph{Editlets}~\cite{DomoszlaiLP14:editlets} are editors with customisable
user interfaces, which can allow multiple users to interact with shared data
sources. Much like incremental relational lenses~\cite{horn2018incremental},
Editlets communicate \emph{changes} in the data as opposed to the
entire data source, however the user must specify this behaviour manually.
\section{Conclusion}\label{sec:conclusion}

Relational lenses allow updatable views of database tables.  Previous work has
concentrated on the semantics of relational lenses, but has not proposed a
concrete language design. As a result, previous implementations imposed severe
limitations on predicates, and provided limited checking of the correctness of
lens composition.

In this paper, we have presented the first full integration of relational lenses
in a functional programming language, by extending the Links
programming language. Building on the approach of~\citet{Cooper09:linq-norm}, we
use normalisation rules to rewrite functional expressions into a form amenable
to compilation to SQL and for typechecking lenses.
Furthermore, we have adapted the existing typing rules for relational lenses to
the setting of a functional programming language and proved that our adapted
rules are sound.

Previous implementations have provided only small example applications. To
demonstrate the use of relational lenses, we have implemented part of the
curation interface for a scientific database as a cross-tier web application,
and shown how relational lenses can be used in tandem with the Model-View-Update
architecture for frontend web development.

As future work, we plan to explore %
integrating relational lenses with auto-incrementing table fields.
\vspace{-0.5em}

\begin{acks}
  Thanks to the anonymous reviewers for their helpful comments and to
  Simon Harding and Jo Sharman for discussion of the GtoPdb curation
  interface. This work was supported by ERC Consolidator Grant Skye
  (grant number \grantnum{ERC}{682315}) and an ISCF Metrology
  Fellowship grant provided by the UK government's Department for
  Business, Energy and Industrial Strategy (BEIS)
\end{acks}

\bibliographystyle{ACM-Reference-Format}
\bibliography{references}

\begin{techreport}
\onecolumn
\allowdisplaybreaks
\appendix
\newpage
\section{Supplementary Material}\label{appendix:supplementary}
\subsection{Functional Dependencies}

Figure \ref{fig:armstrongs_axioms} defines the set of functional dependencies
which can be derived from a set of functional dependencies $\setfdsa$.

\begin{figure}[h!]
  ~\totheleft{Functional Dependencies}
  \begin{flalign*}
    \fdsa, \fdsb \extdef & \slaba \to \slabb
  \end{flalign*}

  ~Armstrong's Axioms \hfill \framebox{$\setfdsa \vDash \slaba \to \slabb$}

  \begin{mathpar}
    \inferrule[FD-ID] { \slaba \to \slabb \in \setfdsa } { \setfdsa \vDash \slaba \to \slabb}

    \inferrule[FD-Transitive] { \setfdsa \vDash \slaba \to \slabb \\ \setfdsa \vDash \slabb
      \to \slabc } { \setfdsa \vDash \slaba \to \slabc }

    \inferrule[FD-Refl.] { \slabb \subseteq \slaba } { \setfdsa \vDash \slaba \to \slabb }

    \inferrule[FD-Aug] { \setfdsa \vDash \slaba \to \slabb } { \setfdsa \vDash \slaba~\slabc \to \slabb~\slabc }

    \inferrule[FD-Composition] { \setfdsa \vDash \slaba \to \slabb \\ \setfdsa
      \vDash \slaba' \to \slabb' } { \setfdsa \vDash \slaba~\slaba' \to \slabb~\slabb' }

    \inferrule[FD-Decomposition] { \setfdsa \vDash \slaba \to \slabb~\slabc } { \setfdsa \vDash \slaba \to \slabb }
  \end{mathpar}
  \caption{Armstrong's functional dependency axioms.}
  \label{fig:armstrongs_axioms}
\end{figure}

\section{Proofs for Section~\ref{sec:impl:predicates}}\label{appendix:predicates}

\begin{fake}{Proposition~\ref{prop:normal-forms}}
If $x : \trowa \vdash M : A$ and $M \norm^* N \not\norm$, then $N$ is in
normal form.
\end{fake}
\begin{proof}

  As the rewrite rules can be applied anywhere in a term, it follows that if we
  cannot apply a normalisation rule to a term, then we cannot apply a
  normalisation rule to any of its subterms.

  Terms typeable by rules \textsc{T-Var}, \textsc{T-Const}, and \textsc{T-Abs}
  are already in normal form. Rules \textsc{T-Record} and \textsc{T-Op} follow
  directly from the induction hypothesis. The remainder of the cases follow.

  \begin{proofcase}{T-App}
    Assumption:
    \begin{mathpar}
      \inferrule*
      { x : \trowa \vdash M : A \to B \\ x : \trowa \vdash N : A }
      { x : \trowa \vdash M \app N : B }
    \end{mathpar}

    By the induction hypothesis, we have that $M$ and $N$ are in normal form.

    Given that $M$ is in normal form and has function type, there are the
    following possibilities:
    \begin{itemize}
      \item $M = x$, which is not possible since the only variable in the
        typing environment is $x$, and $\trowa$ is not a function type
      \item $M = c$, which is not possible since constants only have base types,
        not function types
      \item $M = \lambda x . M$. In this case, we could apply the first
        normalisation rule, which would be a contradiction.
      \item $M = x.\ell$, which is not possible since $\trowa$ only contains
        fields with base types, not function types
      \item $\ifelse{V_1}{V_2}{V_3}$. In this case, we could apply the fifth
          normalisation rule, which would be a contradiction.
      \item $\opargs{\seq{V}}$, which is not possible since the result of an
        operator must have base type.
    \end{itemize}

    Thus, a term $x : \trowa \vdash M \app N$ cannot be in normal form.
  \end{proofcase}

  \begin{proofcase}{T-Project}
    Assumption:
    \begin{mathpar}
      \inferrule*
      { \Gamma \vdash M : \Trecord{\ell_i : A_i}_{i \in I} \\ j \in I }
      { \Gamma \vdash M.\ell_j : A_j }
    \end{mathpar}

    By the IH, we have that $M$ is in normal form.  We now perform case analysis
    on $M$, giving us the following possibilities for terms in normal form which
    can have record type:

    \begin{itemize}
      \item $M = x$: We have that $x.\ell$ which is in normal form.
      \item $M = (\seq{\ell = V})$: Impossible, since it would be possible to
        reduce by the second normalisation rule
      \item $M = \ifelse{V_1}{V_2}{V_3}$: Impossible, since it would be possible
          to reduce by the sixth reduction rule
    \end{itemize}
  \end{proofcase}

  \begin{proofcase}{T-If}

    Assumption:
    \begin{mathpar}
      \inferrule*
      { \Gamma \vdash L : \Tbool \\ \Gamma \vdash M : A \\ \Gamma \vdash N : A }
      { \Gamma \vdash \ifelse{L}{M}{N} : A }
    \end{mathpar}

    Immediate by the induction hypothesis on all three subterms; normalisation
    rules 3 and 4 serve only as an optimisation.
  \end{proofcase}

\end{proof}

\section{Proofs for Section~\ref{sec:impl:predicates}}\label{appendix:typechecking}

\subsection{Predicate Lemmas}

The following definition for set ignores has been taken from \citet{bohannon2006relational}:
\begin{defin}
  $\Pset \preda {\trowa \rowtyapp \trowb}$ ignores $\trowb$ if for all $\rowa,
  \rowb$ if $\rowa[\domain \trowa] = \rowb[\domain \trowa]$ then $\rowa
  \in \Pset \preda {\trowa \rowtyapp \trowb} \iff \rowb \in \Pset \preda {\trowa
    \rowtyapp \trowb}$.
\end{defin}

\begin{lemma}[Set-Ignores]
  Suppose $\preda$ ignores $\domain \trowa$. Then $\Pset \preda \trowa$ ignores
  $\domain \trowb$.
\end{lemma}

\newProofContext
\begin{proof}
  \begin{flalign}
    & \preda \text{ ignores } \domain \trowa & \text{assumption} \notag \\
    & \rowa[\domain \trowa] = \rowb[\domain \trowa] & \text{assumption} \loclabel{same} \\
    & \notag \\
    & \tyto {x : \trowa} \preda \Tbool & \text{def. ignores} \loclabel{typ} \\
    & \notag \\
    & \text{Showing } \rowa
    \in \Pset \preda {\trowa \rowtyapp \trowb} \implies \rowb \in \Pset \preda {\trowa
      \rowtyapp \trowb} & \notag \\
    & \rowa \in \Pset \preda {\trowa \rowtyapp \trowb} & \notag \\
    & \qquad \implies \rowa[\domain \trowa] \in \Pset \preda \trowa &
    \text{weakening}~(\locref{typ}) \notag \\
    & \qquad \implies \rowb[\domain \trowa] \in \Pset \preda \trowa &
    (\locref{same}) \notag \\
    & \qquad \implies \rowb \in \Pset \preda {\trowa \rowtyapp
      \trowb} & \text{widening}~(\locref{typ}) \notag \\
    & \notag \\
    & \text{conversely applies for other direction.} \notag
  \end{flalign}
\end{proof}

\begin{lemma}
  Suppose two disjoint type contexts $\trowa$ and $\trowb$ as well as $\rowa \in
  \inh \trowa$ and $\rowb \in \inh \trowb$. Then $\sat {\rowa \rowapp \rowb}$
  implies $\rowa \in \set { z[\trowa] \mid z \in \inh {\trowa \rowtyapp \trowb}.~
    \sat {z} }$ and \\ $\rowb \in \set { z[\trowb] \mid z \in \inh {\trowa \rowtyapp
      \trowa}.~ \sat {z} }$.
  \label{lemma:sat-implies-in}
\end{lemma}

\begin{proof}
  $\sat{\rowa_1 \rowapp \rowa_2}$ implies that there exists a $\rowb \in \inh{\trowa \rowtyapp
    \trowb}$ which is equal to $\rowa_1 \rowapp \rowa_2$ such that $\sat {\rowb}$. We
  know that for this $\rowb$, $\rowb[\domain {\trowa_1}]$ equals $\rowa_1$ by definition of $\cdot
  \rowapp \cdot$. Conversely the same can be shown for $\rowa_2$.
\end{proof}

\begin{lemma}
  \label{lem:pred-sat-in}
  $\rowa \in \Pset{\preda}{\trowa}$ if and only if $\sat \rowa$.
\end{lemma}

\begin{proof}
  By definition of $\sat \cdot$, $\cdot \in \cdot$ and $\Pset{\cdot}{\cdot}$.
\end{proof}

\subsection{Predicate Equivalences}
\begin{lemma}
  \label{lem:set-inter}
  Suppose two predicates $\preda, \predb$ such that $x : \trowa \vdash \preda$
  and $x : \trowa \vdash \predb$. Then $\Pset{\preda \wedge \predb}{\trowa} =
  \Pset{\preda}{\trowa} \cap \Pset{\predb}{\trowa}$.
\end{lemma}

\newProofContext
\begin{proof}
  \small
  \begin{flalign}
    & x : \trowa \vdash \preda & \text{assumption} \notag \\
    & x : \trowa \vdash \predb & \text{assumption} \notag \\
    & \notag \\
    & \Pset{\preda \wedge \predb}{\trowa} & \notag \\
    & \qquad = \set{\rowa \mid \forall \rowa \in \inh \trowa.~ \sat[\preda
  \wedge \predb] \rowa} & \text{def.}~ \Pset{\cdot}{\cdot} \notag \\
    & \qquad = \set{\rowa \mid \forall \rowa \in \inh \trowa.~ \sat \rowa
\text{ and } \sat [\predb] \rowa} & \text{def.}~ \cdot \wedge \cdot \notag
    \\
    & \qquad = \set{\rowa \mid \forall \rowa \in \inh \trowa.~ \rowa \in
\Pset{\preda}{\trowa} \text{ and } \rowb \in \Pset{\predb}{\trowa}} & \text{Lemma }
    \ref{lem:pred-sat-in} \notag \\
                                                                            &
\qquad = \Pset{\preda}{\trowa} \cap \Pset{\predb}{\trowa} & \text{def.}~\cdot \cap \cdot
    \notag
  \end{flalign}
\end{proof}

\begin{lemma}
  \label{lem:set-join}
  Suppose two predicates $\preda, \predb$ such that $x : \trowa \vdash \preda$
  and $x : \trowb \vdash \predb$. Then $\Pset{\preda \wedge \predb}{\trowa} = 
  \Pset{\preda}{\trowa} \Join \Pset{\predb}{\trowb}$.
\end{lemma}

\newProofContext
\begin{proof}
  \footnotesize
  \begin{flalign}
    & x : \trowa \vdash \preda & \text{assumption} \loclabel{typ-preda} \\
    & x : \trowb \vdash \predb & \text{assumption} \loclabel{typ-predb} \\
    & \notag \\
    & \Pset{\preda \wedge \predb}{\trowa} & \notag \\
    & \qquad = \set{\rowa \mid \forall \rowa \in \inh {\trowa \cup \trowb}.~
      \sat[\preda
  \wedge \predb] \rowa} & \text{def.}~ \Pset{\cdot}{\cdot} \notag \\
    & \qquad = \set{\rowa \mid \forall \rowa \in \inh \trowa.~ \sat \rowa
\text{ and } \sat [\predb] \rowa} & \text{def.}~ \cdot \wedge \cdot \notag
    \\
    & \qquad = \set{\rowa \mid \forall \rowa \in \inh \trowa.~ \sat
      {\rowa[\domain \trowa]} \text{ and } \sat [\predb] {\rowa[\domain
      \trowb]}} &
    (\locref{typ-preda}, \locref{typ-predb}) \notag
    \\
    & \qquad = \set{\rowa \mid \forall \rowa \in \inh \trowa.~ \rowa[\domain \trowa]
      \in \Pset{\preda}{\trowa} \text{ and } \rowb[\domain \trowb] \in
    \Pset{\predb}{\trowb}} &
    \text{Lemma }
    \ref{lem:pred-sat-in} \notag \\
                    & \qquad = \Pset{\preda}{\trowa} \Join
    \Pset{\predb}{\trowb} & \text{def.}~\cdot \Join \cdot
    \notag
  \end{flalign}
\end{proof}

\begin{lemma}
  Suppose $\preda$ such that $\ljd \lampreda \trowa {\set \laba} $ and
  $\defv[\trowa, {\set \laba}]
  \lampreda {\set{\laba = \vala}}$ and $x : (R \cup
  \set{\ell}) \vdash \preda : \Tbool$. Then $\Pset{\subst[x.\laba] \preda
  \vala}{\trowa}
  = \Pset{P}{\trowa}[\domain \trowa]$.
  \label{lem:pred-subst-proj}
\end{lemma}

\newProofContext
\begin{proof}
  \begin{flalign}
    & \Pset{\subst[x.\ell] \preda \vala}{\trowa} &  \notag \\
    & \qquad = \set{ \rowa \mid \forall \rowa \in \inh{\trowa}.~ \sat {\rowa
\uplus \set{\ell = \vala}}} & \text{def.}~\Pset{\cdot}{\cdot} \notag \\
    & \qquad = \set{ \rowa \mid \forall \rowa \in \inh{\trowa}.~\forall \rowb
      \in \inh {\set{\ell}}.~ \sat {\rowa
        \uplus \rowb}} & \defv[\trowa, {\set \laba}] \lampreda {\set{\laba = \vala}},  \ljd \lampreda \trowa {\set \ell} \notag \\
    & \qquad = \set{ \rowa[\domain \trowa] \mid \forall \rowa \in \inh{\trowa \uplus
        \set \ell}.~ \sat {\cup\rowa
      }} & \notag \\
         & \qquad = \Pset{\preda}{\trowa} [\domain \trowa] & \text{def.}~\cdot[\cdot]\notag
  \end{flalign}
\end{proof}

\subsection{Lens Translation}

\begin{figure}[t]
  \begin{mathpar}
    \inferrule[T-Select-RL]{\rlsort R = (U, \setpredb, \setfdsa) \\\\
      \rlsort S = (U, \setpreda \cap \setpredb, \setfdsa)
      \\\\ \setfdsa \text{ is in tree form} \\ \setpredb \text{ ignores } \outputs \setfdsa}
      {\bpvselect{R}{\setpreda}{S} \in \rldisjoint{\set{R}}{\set{S}}}

    \inferrule[T-Join-RL]
    {
      \rlsort{R} = (U, \setpreda, \setfdsa) \\
      \rlsort{S} = (V, \setpredb, \setfdsb) \\\\
      \rlsort{T} = (UV, \setpreda \Join \setpredb, \setfdsa \cup \setfdsb) \\\\
      \setfdsb \vDash U \cap V \to V \\\\
      \setfdsa \text{ is in tree form} \\
      \setfdsb \text{ is in tree form} \\\\
      \setpreda \text{ ignores } \outputs \setfdsa \\ \setpredb \text{ ignores } \outputs \setfdsb
    }
    {\bpvjoin{R}{S}{T} \in \rldisjoint {\set{R, S}} {\set{T}}}

    \inferrule
      [T-Drop-RL]
      {
        \rlsort{R} = (U, \setpreda, \setfdsa) \\\\
        A \in U \\
        \setfdsa \equiv \setfdsb \cup {X \to A} \\\\
        \rlsort{S} = (U-A, \setpreda[U-A], \setfdsb) \\\\
        \setpreda = \setpreda[U-A] \Join \setpreda[A] \\
        \set{A = a} \in \setpreda[A]
      }
      { \bpvdrop{A}{X}{a}{R}{S} \in \rldisjoint {\set{R}} {\set{S}}}

    \inferrule
      [T-Compose-RL]
      { \bpvlensa \in \Sigma \Leftrightarrow \Sigma' \\
       \bpvlensb \in \Sigma' \Leftrightarrow \Delta}
      {\bpvcompose{\bpvlensa}{\bpvlensb} \in \Sigma \Leftrightarrow \Delta}

    \inferrule
      [T-Id-RL]
      { }
      {\lensid \in \Sigma \to \Sigma}
  \end{mathpar}
  \caption{Lens typing rules as defined by \citet{bohannon2006relational}.}
  \label{fig:bpv-lenses}
\end{figure}

\subsubsection{Primitive Lens}

\begin{lemma}
  \label{lem:primitive-wb}
  Suppose $\tyto \tyctx {\lens \expra \setfdsa} {\lensty \trowa \vtrue
    \setfdsa}$. Then $\ltrans {\lens \expra \setfdsa} = \ltriple \lensid \relb$
    and $\rlsort \relb = (\domain {\trowa}, \Pset{\vtrue}{\trowa}, \setfdsa)$.
\end{lemma}

\begin{proof}
  \begin{flalign}
    & \inferrule[T-Lens] { \tyto{\tyctx}{\expra}{\tablety \trowa} \\ \bigcup
      \nodes{\setfdsa} \subseteq \aliases{\trowa} } {
      \tyto{\tyctx}{\lens \expra \setfdsa}{\lensty[\set{\rela}]{\trowa}{\vtrue}{\setfdsa}} } &
    \text{assumption} \notag \\
    & \notag \\
    & \rlsort \rela = (\domain \trowa, \Pset{\vtrue}{\trowa}, \setfdsa) &
    \text{define} \notag \\
    & \lensa = \lensid \in \set{\rela} \Leftrightarrow \set{\rela} &
    \textsc{T-Id-RL} \notag \\
    & \ltrans {\lens \expra \setfdsa} = \ltriple[\set{\rela}] \lensa \rela \notag
  \end{flalign}
\end{proof}

\subsubsection{Select Lens}

\begin{lemma}
  \label{lem:select-wb}
  Suppose $\tyto \tyctx {\lensselect \expra \predb} {\lensty \trowa {\preda
      \wedge \predb} \setfdsa}$ and $\ltrans \expra = \ltriple {\lensa} \rela$
  such that $\lensa \in \schemaa \Leftrightarrow \set{\rela}$ and $\rlsort \rela
  = (\domain \trowa, \Pset{\preda}{\trowa}, \setfdsa)$. Then $\ltrans {\lensselect
    \expra \predb} = \ltriple {\lensb} {\relb}$ such that $\lensb \in \schemaa \Leftrightarrow
  \set{\relb}$ and $\rlsort \relb = (\domain \trowa, \Pset{\preda \wedge
  \predb}{\trowa}, \setfdsa)$.
\end{lemma}

\newProofContext
\begin{proof}
  \begin{flalign}
    & \inferrule{\tyto \tyctx \expra {\lensty \trowa \preda \setfdsa} \\ \tyto
      {x : \trowa} \predb \Tbool \\\\
      \setfdsa \text{ is in tree form} \\
      \preda \text{ ignores }
      \outputs \setfdsa }{\tyto \tyctx {\lensselect \expra \predb}{\lensty \trowa
        {\preda \wedge
          \predb} \setfdsa}} & \text{assumption} \loclabel{types} \\
    & \ltrans \expra = \ltriple \lensa \rela & \text{assumption} \loclabel{M-transl} \\
    & \lensa \in \schemaa \Leftrightarrow \set S & \text{assumption} \loclabel{M-lens} \\
    & \rlsort {\rela} = (\domain \trowa, \Pset{P}{\trowa}, \setfdsa) &
    \text{assumption} \loclabel{sort-R} \\
    \notag \\
    & \rlsort {\relb} = (\domain \trowa, \Pset{P \wedge Q}{\trowa}, \setfdsa) & \text{define}
    \loclabel{sort-S-van} \\
    & \phantom{\rlsort \relb} = (\domain \trowa, \Pset {\preda} \trowa \cap
    \Pset{\predb}{\trowa}, \setfdsa)
    & \text{Lemma \ref{lem:set-inter}} \loclabel{sort-S} \\
     & \Pset{Q}{\trowa} \text{ ignores } \outputs \setfdsa &
    \textsc{Set-Ignores}~(\locref{types})
    \loclabel{Qs-ignores} \\
    \notag \\
     & L_s = \texttt{select from } \rela \texttt{ where } \Pset{\predb}{\trowa} \texttt{
      as } \relb \in \set \rela \Leftrightarrow \set \relb &
    \textsc{T-Select-RL}~(\locref{sort-R}, \locref{sort-S},
    \locref{Qs-ignores}) \loclabel{expr-lens} \\
    & \ltrans {\lensselect \expra \predb} = \ltriple {L; L_s} {\relb} & \text{def.}~ \ltrans \cdot \notag \\
    & L; L_s \in \schemaa \Leftrightarrow \set \relb &
    \textsc{T-Compose-RL}~(\locref{M-lens}, \locref{expr-lens}) \notag
  \end{flalign}
\end{proof}

\subsubsection{Join Lens}

\begin{lemma}
  \label{lem:join-wb}
  Suppose $\tyto \tyctx {\lensjoindl \expra \exprb} {\lensty[\schemaa \uplus
    \schemab] {\trowa \cup \trowb} {\preda \wedge \predb} {\setfdsa \cup
      \setfdsb}}$ and $\ltrans M = \ltriple[\schemaa] {\lensa} {\rela_1}$ and
  $\ltrans N = \ltriple[\schemab] {\lensb} {\rela_2}$ such that $\lensa_1 \in
  \schemaa \Leftrightarrow \set {\rela_1}$, $\rlsort {\rela_1} = (\domain
  \trowa, \Pset{\preda}{\trowa}, \setfdsa)$ and $\lensa_2 \in \schemab \Leftrightarrow
  \set {\rela_2}$, $\rlsort {\rela_2} = (\domain \trowb, \Pset{\predb}{\trowb},
  \setfdsb)$. Then $\ltrans {\lensjoindl \expra \exprb} = \ltriple[\schemaa
  \uplus \schemab] \lensa \relb$ such that $\lensa \in {\schemaa \uplus
    \schemab} \Leftrightarrow \set \relb$ and $\rlsort \relb = (\domain {\trowa
  \rowtyapp \trowb}, \Pset{\preda \wedge \predb}{\trowa \rowtyapp \trowb}, \setfdsa \cup \setfdsb)$.
\end{lemma}

\newProofContext
\begin{proof}
  \begin{flalign}
    & \inferrule[T-Join-Left]{ \tyto{\tyctx}{\expra}{\lensty \trowa \preda \setfdsa} \\
      \tyto{\tyctx}{\exprb}{\lensty[\schemab] \trowb \predb \setfdsb} \\\\
      \setfdsb \vDash \aliases \trowa \cap \aliases \trowb \to \aliases \trowb \\
      \setfdsa \text{ is in tree form} \\
      \setfdsb \text{ is in tree form} \\\\
      \ignoresoutputs{\preda}{\setfdsa} \\
      \ignoresoutputs{\predb}{\setfdsb} \\
      \schemaa \cap \schemab = \emptyset } {
      \tyto{\tyctx}{\lensjoindl{\expra}{\exprb}}{\lensty[\schemaa \cup \schemab]
        {\trowa \cup
          \trowb} {\preda \wedge \predb} {\setfdsa \cup \setfdsb}}} & \text{assumption} \loclabel{typ-out} \\
    & \ltrans \expra = \ltriple {\lensa_1} {\rela_1} & \text{assumption} \loclabel{lens-M} \\
    & \ltrans \exprb = \ltriple {\lensa_2} {\rela_2} & \text{assumption} \loclabel{lens-N} \\
    & \lensa_1 \in \schemaa \Leftrightarrow \set {\rela_1} & \text{assumption} \\
    & \lensa_2 \in \schemab \Leftrightarrow \set {\rela_2} & \text{assumption} \\
    & \rlsort {\rela_1} = (\domain \trowa, \Pset{\preda}{\trowa}, \setfdsa) &
    \text{assumption} \loclabel{sort-R1} \\
    & \rlsort {\rela_2} = (\domain \trowb, \Pset{\predb}{\trowb}, \setfdsb) &
    \text{assumption} \loclabel{sort-R2} \\
    & \notag \\
    & \lensa_1 \in \schemaa \uplus \schemab \Leftrightarrow \set {\rela_1}
    \uplus \schemab &
    \text{weakening}^* \\
    & \lensa_2 \in \set {\rela_1} \uplus \schemab \Leftrightarrow \set {\rela_1,
      \rela_2} &
    \text{weakening}^* \\
    & \lensa_1 ; \lensa_2 \in \schemaa \uplus \schemab \Leftrightarrow \set{\rela_1, \rela_2} & \textsc{T-Compose-RL}  \notag \\
    & \qquad {}^* \text{ all intermediate views are globally unique due to def.
      of $\ltrans \cdot$.}
    \notag \\
    & \notag \\
    & \rlsort {\relb} = (\domain {\trowa \rowtyapp \trowb}, \Pset{\preda \wedge
    \predb}{\trowa \rowtyapp \trowb}, \setfdsa \cup \setfdsb) & \text{define} \loclabel{sort-S-van} \\
                                                              & 
    \phantom{\rlsort {\relb}} = (\domain\trowa \rowtyapp \domain \trowb,
    \Pset{\preda}{\trowa}
    \Join \Pset \predb \trowb, \setfdsa \cup \setfdsb) & \text{Lemma \ref{lem:set-join}} \loclabel{sort-S} \\
     & \Pset{\preda}{\trowa} \text{ ignores } \outputs \setfdsa &
    \textsc{Set-Ignores}~(\locref{typ-out}) \loclabel{ignores-Pset} \\
     & \Pset{\predb}{\trowb} \text{ ignores } \outputs \setfdsb &
    \textsc{Set-Ignores}~(\locref{typ-out}) \loclabel{ignores-Qset} \\
    & \notag \\
    & & \textsc{T-Join-RL} \notag \\
    & L_j = \texttt{join\_dl } \rela_1 , \rela_2 \texttt{ as } \relb \in
    \set{\rela_1, \rela_2} \Leftrightarrow \set{\relb} & (\locref{sort-R1},
    \locref{sort-R2}, \locref{sort-S}, \locref{typ-out}
    \locref{ignores-Pset}, \locref{ignores-Qset}) \notag \\
    & \ltrans {\lensjoindl \expra \exprb} = \ltriple[\schemaa \uplus \schemab]
    {\lensa_1; \lensa_2;
      \lensa_j} \relb & \text{def.}~\ltrans \cdot \notag \\
    & \lensa_1 ; \lensa_2 ; \lensa_j \in \schemaa \uplus \schemab
    \Leftrightarrow \set{\relb} & \textsc{T-Compose-RL} \notag
  \end{flalign}
\end{proof}

\subsubsection{Drop Lens}
\label{section:appendix:drop-lens}

\begin{fake}{Lemma~\ref{lem:ljd-sound}}
  Suppose $\trowa = \trowa_1 \rowtyapp \trowa_2$ and $x : \trowa \vdash \preda :
  \Tbool$. If
  $\ljd{\lampreda}{\trowa_1}{\trowa_2}$, then $\Pset{\preda}{\trowa} =
  \Pset{\preda}{\trowa}[\domain
  {\trowa_1}] \Join \Pset{\preda}{\trowa}[\domain {\trowa_2}]$.
\end{fake}

\newProofContext
\begin{proof}
  \begin{flalign}
    & \trowa = \trowa_1 \uplus \trowa_2 & \text{assumption} \notag \\
    & \ljd \lampreda {\trowa_1} {\trowa_2} & \text{assumption} \notag \\
    \notag \\
    & \forall \rowa, \rowb \in \inh {\trowa}.~ \sat \rowa \text{ iff. } & \notag \\
    & \qquad \sat{\rowa[\domain {\trowa_1}] \rowapp \rowb[\domain {\trowa_2}]}
    \wedge \sat{\rowa[\domain {\trowa_2}] \rowapp \rowb[\domain {\trowa_1}]}
    & \text{def.}~\ljd \lampreda {\trowa_1} {\trowa_2} \loclabel{implies} \\
    \notag \\
    & \Pset{\preda}{\trowa} & \notag \\
    & \qquad = \set{ \rowa \mid \forall \rowa \in \inh \trowa.~ \sat \rowa} &
    \text{def.}~
    \Pset{\cdot}{\cdot} \notag \\
    & \qquad = \set{ \rowa \mid \forall \rowa,\rowb \in \inh
      \trowa.~\sat{\rowa[\domain {\trowa_1}] \rowapp \rowb[\domain {\trowa_2}]}
      \wedge \sat{\rowa[\domain {\trowa_2}]
        \rowapp \rowb[\domain {\trowa_1}]}} & (\locref{implies}) \notag \\
    & \qquad = \{ \rowa \mid \forall \rowa \in \inh \trowa.~\rowa[\domain
    {\trowa_1}] \in \set{\rowb[\domain {\trowa_1}] \mid \forall \rowb \in
      \inh{\trowa}.~ \sat \rowb} \text{
      and } & \notag \\
    & \qquad \qquad \qquad \rowa[\domain {\trowa_2}] \in \set{\rowb[\domain
      {\trowa_2}] \mid \forall \rowb \in \inh
      \trowa.~ \sat \rowb} \} & \text{Lemma } \ref{lemma:sat-implies-in} \notag \\
    & \qquad = \set{ \rowa \mid \forall \rowa \in \inh \trowa .~\rowa[\domain {\trowa_1}] \in \Pset
      {\preda}{\trowa}[\trowa_1] \text{ and } \rowa[\trowa_2] \in \Pset
    {\preda}{\trowa}[\trowa_2]} & \text{def. } \cdot[\cdot] \notag \\
                                & \qquad = \Pset{\preda}{\trowa}[\domain {\trowa_1}] \Join
                                \Pset{\preda}{\trowa}[\domain {\trowa_2}]
    & \text{def. } \cdot \Join \cdot \notag
  \end{flalign}
\end{proof}

\begin{fake}{Lemma~\ref{lem:ljdi-sound}}
  Given a predicate $\lampreda$ and record types $\trowa, \trowb$, it follows that
  $\ljdi \lampreda \trowa \trowb$ implies $\ljd \lampreda \trowa \trowb$.
\end{fake}

\newProofContext
\begin{proof}
  \begin{flalign}
    & \rowa_1, \rowa_2 \in \inh \trowa & \text{assumption} \notag \\
    & \rowb_1, \rowb_2 \in \inh \trowb & \text{assumption} \notag \\
    & \sat {\rowa_1 \rowapp \rowb_1} & \text{assumption} \loclabel{sat-1} \\
    & \sat {\rowa_2 \rowapp \rowb_2} & \text{assumption} \loclabel{sat-2} \\
    & \notag \\
    & \text{Perform induction on } \ljdi \lampreda \trowa \trowb & \notag \\
    & \notag \\
    & \inferrule[LJD$^\dagger$-1]
    {\tyto {x : \trowa} \preda \Tbool }{\ljdi{\lampreda}{\trowa}{\trowb}} & \text{assumption} \notag \\
    & \qquad \sat {\rowa_1} & \text{Extensionality}~(\locref{sat-1}) \notag \\
    & \qquad \sat {\rowa_1 \rowapp \rowb_2} & \text{Extensionality} \notag \\
    & \qquad \ljd \lampreda \trowa \trowb & \text{def. LJD} \notag \\
    & \notag \\
    & \inferrule [LJD$^\dagger$-2] { \tyto {x : \trowb} \preda \Tbool }
    {\ljdi{\lampreda}{\trowa}{\trowb}} & \text{assumption} \notag \\
    & \qquad \sat {\rowb_2} & \text{Extensionality}~(\locref{sat-2}) \notag \\
    & \qquad \sat {\rowa_1 \rowapp \rowb_2} & \text{Extensionality} \notag \\
    & \qquad \ljd \lampreda \trowa \trowb & \text{def. LJD} \notag \\
    & \notag \\
    & \inferrule
    [LJD$^\dagger$-And]
    {\ljdi{\lamf {\predb_1}}{\trowa}{\trowb} \\ \ljdi{\lamf {\predb_2}}{\trowa}{\trowb} }
    {\ljdi {\lamf {\predb_1 \wedge \predb_2}}{\trowa}{\trowb}}
    & \text{assumption} \\
    & \qquad \sat[\predb_1]{\rowa_1 \rowapp \rowb_2} & \text{induction} \notag \\
    & \qquad \sat[\predb_2]{\rowa_1 \rowapp \rowb_2} & \text{induction} \notag \\
    & \qquad \sat[\predb_1 \wedge \predb_2]{\rowa_1 \rowapp \rowb_2} &
    \text{def.}~\cdot \wedge \cdot \notag \\
    & \qquad \ljd {\lamf {\predb_1 \wedge \predb_2}} \trowa \trowb & \text{def. LJD} \notag \\
  \end{flalign}
\end{proof}

\begin{fake}{Lemma~\ref{lem:defv-sound}}
  Suppose $\trowa = \trowa_1 \uplus \trowa_2$, $\trowa_2 \vdash \rowa$ and $\ljd
  \lampreda {\trowa_1} {\trowa_2}$. Then $\defv[\trowa_1, \trowa_2] \lampreda \rowa$ implies $\rowa \in
  \Pset{\preda}{\trowa}[\domain {\trowa_2}]$.
\end{fake}

\newProofContext
\begin{proof}
  \begin{flalign*}
    & \trowa = \trowa_1 \uplus \trowa_2 & \text{assumption} \notag \\
    & \notag \\
    & \exists \rowb \in \inh{\trowa_1}.~ \sat {\rowa \rowapp \rowb} & \text{assumption} \notag \\
    & \implies \exists \rowb \in \inh{\trowa}.~ \rowb[\domain \trowa_2] = \rowa
    \text{
      and } \sat \rowb & \text{def. } \cdot \rowapp \cdot \notag \\
    & \implies \rowa \in \set{\rowb[\domain {\trowa_2}] \mid \forall \rowb \in
      \inh{\trowa}.~\sat \rowa} &
    \text{def. } \cdot \in \cdot \notag \\
                                & \implies \rowa \in \Pset{\preda}{\trowa}[\domain {\trowa_2}] & \text{def. } \cdot[\cdot] \notag
  \end{flalign*}
\end{proof}

\begin{fake}{Lemma~\ref{lem:defvi-sound}}
  Given a predicate $\lampreda$ such that $\Pset \preda {\trowa \rowtyapp
    \trowb}$ is not empty and record $\rowa$ of type $\trowa$, it follows that $\defvi{\lampreda}{\rowa}$
  implies $\defv{\lampreda}{\rowa}$.
  \label{lem:defvi-sound-proof}
\end{fake}

\newProofContext
\begin{proof}
  \begin{flalign}
    & \rowa \in \inh \trowb & \text{assumption} \notag \\
    & \Pset \preda {\trowa \rowtyapp \trowb} \text{ not empty} &
    \loclabel{not-empty} \\
    & \notag \\
    & \text{Perform induction on } \defvi \lampreda \rowa & \notag \\
    &  \inferrule[DV$^\dagger$-1]
    { \tyto {x : \trowa} {\preda} {D} }
    {\defvi{\lampreda}{\rowa}}
    & \text{assumption} \notag \\
    & \qquad \exists \rowb \in \inh \trowa.~ \sat{\rowb} & (\locref{not-empty})
    \notag \\
    & \qquad \exists \rowb \in \inh \trowa.~ \sat{\rowa \rowapp \rowb} &
    \text{widening} \notag \\
    & \qquad \defv \lampreda \rowa & \text{def. DV} \notag \\
    & \notag \\
    & \inferrule
    [DV$^\dagger$-2]
    { \tyto {x : \trowb} {\preda} {D} \\\\ \sat \rowa }
    {\defvi{\lampreda}{\rowa}} & \notag \\
    & \qquad \exists \rowb \in \inh \trowa.~ \sat{\rowa \rowapp \rowb} &
    \text{widening} \notag \\
    & \qquad \defv \lampreda \rowa & \text{def. DV} \notag \\
    & \notag \\
    & \inferrule
    [DV$^\dagger$-And]
    { \defvi{\lamf {\predb_1}}{\rowa} \\\\ \defvi{\lamf {\predb_2}}{\rowa}}
    {\defvi{\lamf {\predb_1 \wedge \predb_2}}{\rowa}} & \text{assumption} \notag \\
    & \qquad \exists \rowb \in \inh \trowa.~ \sat[\predb_1]{\rowa \rowapp \rowb}
    & \text{induction} \notag \\
    & \qquad \exists \rowb \in \inh \trowa.~ \sat[\predb_2]{\rowa \rowapp \rowb}
    & \text{induction} \notag \\
    & \qquad \exists \rowb \in \inh \trowa.~ \sat[\predb_1 \wedge \predb_2]{\rowa \rowapp \rowb}
    & \text{def.}~\cdot \wedge \cdot \notag \\
    & \qquad \defv \lampreda \rowa & \text{def. DV} \notag
  \end{flalign}
\end{proof}

\begin{lemma}
  \label{lem:drop-wb}
  Suppose $\tyto \tyctx {\lensdrop {\laba'} \slaba v \expra} {\lensty \trowa
    \predb \setfdsb}$ and $\ltrans \expra = \ltriple \lensa \rela$ such that
  $\lensa \in \schemaa \Leftrightarrow \set \rela$ and $\rlsort \rela = (\domain
  \trowa \cup \laba', \preda, \setfdsa)$. Then $\ltrans {\lensdrop {\laba'}
    \slaba v \expra} = \ltriple {\lensb} \relb$ such that $\lensb \in \schemaa
  \Leftrightarrow \set \relb$ and $\rlsort \relb = (\domain \trowa, 
  \Pset{\predb}{\trowa}, \setfdsb)$.
\end{lemma}

\newProofContext
\begin{proof}
  \begin{flalign}
    & \inferrule[T-Drop]{ \setfdsa \equiv \setfdsb \cup \set{\slaba \to \laba'} \\
      \tyto{\tyctx}{\expra}{\lensty {\trowa \cup \set{\oftype{\laba'}{\tau}}}
        \preda
        \setfdsa} \\\\
      \slabb = \aliases \trowa \\ \tyto\tyctx \vala A\\\\
      \ljd \lampreda \trowa {\set{\laba' : \typpa}} \\
      \defv[\trowa,\set{\laba' : A}] \lampreda {\set {\laba' = \vala}} } {
      \tyto{\tyctx}{\lensdrop {\laba'} \slaba \vala \expra}{\lensty \trowa
        {\subst[x.\laba'] \preda \vala} \setfdsb} } & \text{assumption}
    \loclabel{typ-out} \\
    & \ltrans \expra = \ltriple \lensa \rela & \text{assumption} \loclabel{lens-M} \\
    & \lensa \in \schemaa \Leftrightarrow \set \rela & \text{assumption} \notag \\
    & \rlsort \rela = (\domain \trowa \cup \laba', \preda, \setfdsa) &
    \text{assumption} \loclabel{sort-R} \\
    & \notag \\
    & \Pset{\preda}{\trowa} = \Pset{\preda}{\trowa}[\domain{\trowa}] \Join \Pset{\preda}[\laba']
    &
    \text{Lemma \ref{lem:ljd-sound}} \loclabel{set-ljd} \\
    & \set{\laba' = \vala} \in \Pset{\preda[\laba']}{\trowa} & \text{Lemma \ref{lem:defv-sound}} \loclabel{defv} \\
    & \rlsort \relb = (\domain \trowa, \Pset {\subst[x.\laba'] \preda \vala}
    {\trowa \cup \set{\laba' : \typpa}}, \setfdsb) & \text{define} \loclabel{sort-S} \\
    & \phantom{\rlsort \relb} = (\domain {\trowa},
    \Pset{\preda}{\trowa \cup \set{\laba' : \typpa}}[\domain{\trowa}],
    \setfdsb) & \text{Lemma \ref{lem:pred-subst-proj}} \loclabel{sort-S-RL} \\
    & \set{\laba'} \in \domain {\trowa \cup \set{\laba'}} & \text{trivial} & \loclabel{A-in-U} \\
    & \notag \\
    & \lensa_d = \texttt{drop } \laba \texttt{ determined by } (\slaba,\vala)
    \texttt{ from } \rela \texttt{ as } \relb \in \set{\rela} \Leftrightarrow
    \set{\relb} & \textsc{T-Drop-RL} (\locref{sort-R}, \locref{A-in-U},
    \locref{typ-out},
    \locref{sort-S-RL}, \locref{set-ljd}, \locref{defv}) \loclabel{ld-type} \\
    & \ltrans {\lensdrop {\laba'} \slaba v \expra } = \ltriple {L; L_d} \relb &
    \text{def.}~\ltrans \cdot\notag \\
    & L; L_d \in \schemaa \Leftrightarrow \set{\relb} &
    \textsc{T-Compose-RL}~(\locref{ld-type}) \notag
  \end{flalign}
\end{proof}

\subsection{Translation Sound}
\begin{theorem}
  Suppose $\tyto \tyctx \expra {\lensty \trowa \preda \setfdsa}$ and $\ltrans
  \expra = \ltriple \lensa \relb$ then $\lensa \in \schemaa \Leftrightarrow
  \set{\relb}$ and $\rlsort \relb = (\domain \trowa, \Pset{\preda}{\trowa}, \setfdsa)$.
\end{theorem}

\newProofContext
\begin{proof}
  \begin{flalign}
    & \text{perform induction on } \tyto \tyctx \expra {\lensty \trowa \preda \setfdsa} & \notag \\
    & \notag \\
    & \tyto \tyctx {\lens \expra \setfdsa} {\lensty [\set{\relb}] \trowa \vtrue \setfdsa}
    & \text{assumption} \loclabel{lens-M-typ} \\
    & \qquad \ltrans {\lens \expra \setfdsa} = \ltriple[\set{\relb}] \lensa \relb &
    \text{Lemma \ref{lem:primitive-wb} (\locref{lens-M-typ})} \notag \\
    & \qquad \lensa \in \set{\relb} \Leftrightarrow \set{\relb} & 
    \text{Lemma \ref{lem:primitive-wb} (\locref{lens-M-typ})} \notag \\
    & \notag \\
    & \tyto \tyctx {\lensselect \expra \predb} {\lensty \trowa {\preda \wedge
        \predb} \setfdsa} & \text{assumption} \notag \\
    & \qquad \tyto \tyctx \expra {\lensty \trowa \preda \setfdsa} & \textsc{T-Select}
    \loclabel{sel-M-typ} \\
    & \qquad \ltrans \expra = \ltriple \lensa \rela & \text{def.}~\ltrans \cdot \notag \\
    & \qquad \lensa \in \schemaa \Leftrightarrow \set \rela & \text{by induction}
    \loclabel{sel-R-sort} \\
    & \notag \\
    & \qquad \ltrans {\lensselect \expra \predb} = \ltriple {\lensa'} \relb & \text{Lemma
      \ref{lem:select-wb} (\locref{sel-M-typ},
      \locref{sel-R-sort})} \notag \\
    & \qquad \lensa' \in \schemaa \Leftrightarrow \set \relb & \text{Lemma \ref{lem:select-wb}
      (\locref{sel-M-typ},
      \locref{sel-R-sort})} \notag \\
    & \notag \\
    & \tyto \tyctx {\lensjoindl \expra \exprb} {} {\lensty[\schemaa \uplus
      \schemab] {\trowa \cup \trowb}
      {\preda
        \wedge \predb} {\setfdsa \cup \setfdsb}} & \text{assumption} \notag \\
    & \qquad \tyto \tyctx \expra {\lensty[\schemaa] \trowa \preda \setfdsa} &
    \textsc{T-Join}
    \loclabel{join-typ} \\
    & \qquad \ltrans \expra = \ltriple[\schemaa] {\lensa_1} {\rela_1}
    & \text{def.}~\ltrans \cdot \notag \\
    & \qquad \lensa_1 \in \schemaa \Leftrightarrow \rela_1 & \text{by induction}
    \loclabel{join-R1-sort} \\
    & \qquad \ltrans \exprb = \ltriple[\schemab] {\lensa_2} \set{\rela_2} &
    \text{def.}~\ltrans \cdot \notag \\
    & \qquad \lensa_2 \in \schemab \Leftrightarrow \set{\rela_2} & \text{by
      induction} \loclabel{join-R2-sort} \\
    & \notag \\
    & \qquad \ltrans {\lensjoindl \expra \exprb} = \ltriple [\schemaa \uplus \schemab]
    \lensa \relb & \text{Lemma \ref{lem:join-wb} (\locref{join-typ},
      \locref{join-R1-sort}, \locref{join-R2-sort})} \notag \\
    & \qquad \lensa \in \schemaa \uplus \Delta \Leftrightarrow \set{\relb} & \text{Lemma \ref{lem:join-wb} (\locref{join-typ},
      \locref{join-R1-sort}, \locref{join-R2-sort})} \notag \\
    & \notag \\
    & \tyto \tyctx {\lensdrop {\laba'} \slaba v \expra} {\lensty \trowa \predb
      \setfdsa} & \text{assumption} \notag \\
    & \qquad \tyto \tyctx \expra {\lensty {\trowb} \predb \setfdsb} &
    \textsc{T-Drop} \loclabel{drop-typ} \\
    & \qquad \ltrans \expra = \ltriple {\lensa} {\relb} & \text{def.}~\ltrans
    \cdot \notag \\
    & \qquad \lensa \in \schemaa \Leftrightarrow \set{\relb} & \text{by induction} \loclabel{drop-R-sort} \\
    & \notag \\
    & \qquad \ltrans{\lensdrop {\laba'} \slaba v \expra} = \ltriple \lensa \relb
    & \text{Lemma
      \ref{lem:drop-wb} (\locref{drop-typ}, \locref{drop-R-sort})} \notag \\
    & \qquad L \in \schemaa \Leftrightarrow \set{\relb} 
    & \text{Lemma
      \ref{lem:drop-wb} (\locref{drop-typ}, \locref{drop-R-sort})} \notag
  \end{flalign}
\end{proof}

 \end{techreport}

\end{document}